\documentclass[journal,12pt,onecolumn,draftclsnofoot,USletter]{IEEEtran}

\usepackage{url}
\usepackage{cite}
\usepackage{times}
\usepackage{xcolor}
\usepackage{graphicx}
\usepackage{booktabs}
\usepackage{amssymb,amsfonts,amsthm,bm,bbm,dsfont,amsmath}
\usepackage{mathtools}
\usepackage{textcomp}
\newtheorem{theorem}{Theorem}
\newtheorem{remark}{Remark}

\newtheorem{proposition}{Proposition}
\newtheorem{example}{Example}
\newtheorem{lemma}{Lemma}

\title{\LARGE DNA Tails for Molecular Flash Memory} 
\author{%
  \IEEEauthorblockN{\normalsize\textbf{Jin Sima}\IEEEauthorrefmark{1},
                    \textbf{Chao Pan}\IEEEauthorrefmark{2},
                    \textbf{S. Kasra Tabatabaei}\IEEEauthorrefmark{3},
                    \textbf{Alvaro G. Hernandez}\IEEEauthorrefmark{4},
                    \textbf{Charles M. Schroeder}\IEEEauthorrefmark{5}\IEEEauthorrefmark{6}\IEEEauthorrefmark{7}
                    and \textbf{Olgica Milenkovic}\IEEEauthorrefmark{1}\IEEEauthorrefmark{8}
                    \\
                    }
  \IEEEauthorblockA{\IEEEauthorrefmark{1}%
                   \small{Department of Electrical and Computer Engineering, University of Illinois Urbana-Champaign, }\texttt{\{jsima,milenkov\}@illinois.edu}\\}
  \IEEEauthorblockA{\IEEEauthorrefmark{2}%
                    Google, \texttt{chaopan@google.com}\\}
  \IEEEauthorblockA{\IEEEauthorrefmark{3}%
                    New England BioLabs, \texttt{stabatabaei@neb.com}\\}
  \IEEEauthorblockA{\IEEEauthorrefmark{4}%
                    Roy J. Carver Biotechnology Center, University of Illinois Urbana-Champaign, \texttt{aghernan@illinois.edu}\\}
  \IEEEauthorblockA{\IEEEauthorrefmark{5}%
                    Center for Biophysics and Quantitative Biology, University of Illinois Urbana-Champaign, \texttt{cms@illinois.edu}\\}
  \IEEEauthorblockA{\IEEEauthorrefmark{6}%
                    Beckman Institute for Advanced Science and Technology, University of Illinois Urbana-Champaign\\}
  \IEEEauthorblockA{\IEEEauthorrefmark{7}%
                    Department of Materials Science and Engineering, University of Illinois Urbana-Champaign\\}
  \IEEEauthorblockA{\IEEEauthorrefmark{8}
                    Center for Artificial Intelligence and Modeling, Carl R. Woese Institute for Genomic Biology, University of Illinois Urbana-Champaign}
}
\begin{document} 
\baselineskip24pt

\maketitle 
\vspace{-0.25in}

\begin{abstract}
  DNA-based data storage systems face practical challenges due to the high cost of DNA synthesis. A strategy to address the problem entails encoding data via topological modifications of the DNA sugar-phosphate backbone. The DNA Punchcards system, which introduces nicks (cuts) in the DNA backbone, encodes only one bit per nicking site, limiting density. We propose \emph{DNA Tails,} a storage paradigm that encodes nonbinary symbols at nicking sites by growing enzymatically synthesized single-stranded DNA of varied lengths. The average tail lengths encode multiple information bits and are controlled via a staggered nicking-tail extension process. We demonstrate the feasibility of this encoding approach experimentally and identify common sources of errors, such as calibration errors and stumped tail growth errors. To mitigate calibration errors, we use rank modulation proposed for flash memory. To correct stumped tail growth errors, we introduce a new family of rank modulation codes that can correct ``stuck-at'' errors. Our analytical results include constructions for order-optimal-redundancy permutation codes and accompanying encoding and decoding algorithms. 
\end{abstract}
% \begin{sciabstract}
%   Teaser: Meet DNA Tails, a breakthrough approach that harnesses DNA's topological dimension for unparalleled storage density and security.
% \end{sciabstract}
\section{Introduction}
DNA-based data storage systems provide distinct advantages over conventional magnetic, optical, and flash storage media in terms of data storage density, data longevity, and energy efficiency~\cite{church2012next,goldman2013towards,grass2015robust,yazdi2015dna1}. They also offer random-access and rewriting solutions, made possible through controlled polymerase chain reaction (PCR) and overlap-extension PCR reactions~\cite{yazdi2015rewritable1}, or specialized microelectronic circuitry~\cite{khandelwal2022self}. 
The systems can be made portable through the use of nanopore sequencers~\cite{yazdi2017portable}, and adapted to write and read using chemically modified DNA~\cite{tabatabaei2022expanding}. Nevertheless, they still have not been broadly adopted due to substantial implementation challenges such as the high cost of DNA synthesis.

One strategy to mitigate the use of expensive synthetic DNA is to create topological modifications on native DNA backbones to encode user-defined information. The first known system to use topological modifications of the form of nicks (cuts) in one of the backbones of the double-helix is \emph{DNA Punchcards}~\cite{tabatabaei2020dna,pan2022rewritable}. However, DNA Punchcards encode only a single bit of information at each nicking site, thereby offering only a fraction of the recording density achievable by sequence-content storage mechanisms. To bridge the gap between the storage densities of DNA Punchcards and sequence-based storage systems, one needs to find a way to increase the alphabet size available for storing information at the nicking sites. We hence propose to encode nonbinary information at the nicking sites by using an approach inspired by classical flash memory where cell charges represent nonbinary values. We refer to our new approach as \emph{DNA Tails}, since nonbinary information at each nicking site is recorded via enzymatically synthesized single-stranded ``tails,'' whose \textbf{quantized average lengths} represent multiple bits of information. The challenge of controlling the ranges of lengths of the enzymatically synthesized DNA tails is addressed through a staggered nicking-tail extension approach and the use of rank modulation coding~\cite{jiang2009rank}. With this design, the average tail lengths are dictated by the time at which their corresponding sites were nicked. 

We implement the DNA Tail system and test it experimentally. The experiments show that a common source of errors is that tails unexpectedly stop growing after a certain number of rounds of extensions, which we call "stumped" tails. As a result, the information sequence carried by DNA tails suffers from "stuck-at" errors, where some symbols get stuck at lower, incorrect values.  We consider three models of "stuck-at" error scenarios, where: (a) $t$ symbols get stuck at a value lower by $1$ than intended; (b) $t$ consecutive symbols get stuck at the lowest of their values; (c) a single symbol gets stuck at a lower value, and only relative rankings of the remaining symbol values are observed. 

We propose new code constructions and encoding/decoding algorithms for each of the three error models. Our codes for model (a) use Lehmer encoding, which was also used in~\cite{barg2010codes} for classical rank modulation coding. For model (b), we propose a code where the permutation is split into subblocks based on symbol values. Moreover, we use two interleaved splits of the permutation to correct errors. Our codes for model (c) may be viewed as codes correcting a constrained combination of a single deletion and a single erasure in a nonbinary sequence, which is based on a generalization of the Tenengolts codes for correcting a single deletion in nonbinary sequences. We also use Lehmer encoding tailored to permutations. We complement all the constructions with encoding/decoding algorithms that transform information strings into permutations and vice-versa. 

The paper is organized as follows. Section~\ref{sec:experiments} describes the system and experimental results that motivate our analysis. Section~\ref{sec:errormodels} contains a description of the error models, while  code constructions and encoding/decoding algorithms are presented in Section~\ref{sec:codes}.

\section{Experimental System Design and Error Analysis}\label{sec:experiments}

The gist of our approach is to encode nonbinary symbols (labels) using different lengths of single-stranded DNA strings grown at specific nicking (cutting) sites of double-stranded DNA. The sites at which tails are grown are termed nicking sites, while the overall storage paradigm is henceforth referred to as DNA Tails. For DNA, the sugar-phosphate backbone locations naturally serves as a linear order for the encoded symbols.  Following this idea, we designed and implemented a DNA Tail scheme as depicted in Figure~\ref{fig:general} (a), where the tails are single-stranded DNA fragments enzymatically synthesized on double-stranded substrates. The writing process consists of several rounds of enzymatic nicking at preselected locations (indicated by light green crosses on the DNA duplexes and marked by $0$s in the top row of Figure~\ref{fig:general} (a)) and ``labeling.'' The results are tails whose different random lengths represent different symbols of a large coding alphabet\footnote{Note that, as for sequence-based encoding, pools of $100$s of DNA strings encode the same information; since in our case, the tail lengths are random variables, we work with the average tail lengths for each designated nicking location. Furthermore, we quantize the tail lengths in order to allow for a range of tail-length values to represent the same symbol}. Label $0$ stands for an undisturbed location (no nick, nor tail), label $1$ stands for a nicked location without a tail, while all larger labels (e.g., $2$-$7$) correspond to average lengths of the DNA tails of different lengths. The smaller the label, the shorter the average length of the tail. To control the relative average tail lengths, we partition the locations of the tails to be synthesized according to the value of the label, in decreasing order. For example, to ``$70216745$'' at consecutive preselected locations, we start with the two locations that are to store $7$. These are the first locations that we nick, as indicated in the second row of Figure~\ref{fig:general} (b). Upon this first nicking round, DNA tails are grown under controlled conditions, leading to relatively short tail lengths. The locations where the symbol $6$ appears are nicked next, followed by enzymatic synthesis at all ``exposed'' sites -- i.e., those that are nicked and those that contain tails. Since the sites corresponding to label $7$ are subject to two rounds of enzymatic synthesis, their lengths are expected, on average, to be longer than those of label $6$, as illustrated in the fifth row of the figure. Proceeding, we arrive at the construct in the sixth row in which the tails are of different lengths proportional to the symbol value to be stored. %Note that our encoding procedure involves a staggered nicking-tail extension process, as in this case, it is easier for us to control the relative average tail lengths extended at each nicking site, compared to directly extending the tails to desired lengths.

However, as will become evident from the experiments, relying on the exact values of average tail lengths to determine the encoded symbol is often infeasible due to calibration errors (i.e., not knowing very precisely which tail length rages correspond to which symbols). On the other hand, the staggered nicking-tail extension process naturally guarantees that the nicked sites exposed to the most tail extension rounds will have the longest DNA tails with high probability. This motivates us to use relative ordering of the average tail lengths rather than their values, akin to rank modulation~\cite{jiang2009rank,farnoud2012rank},  illustrated in Figure~\ref{fig:general} (c,d). There, to avoid absolute errors,  the exact values are replaced by rank-ordered symbols indicating the largest, second largest, third largest, etc., charge or tail length. Even after charge leakage of all cells or equally reduced tail growths, one still expects the relative order to be preserved. To understand how to implement a rank modulation-like encoding with DNA Tails, one can think of replacing electrons and cells with bases and nicking sites, in which case each average tail length has to be sufficiently different from any other. This makes the recording process less susceptible to errors encountered in the general scheme but at the cost of an increased number of nicking-labeling rounds. For the general scheme, the number of rounds equals the value of the largest label, while for rank modulation scheme, the number of rounds equals the number of distinct nonzero labels.

\begin{figure}
\centering
\includegraphics[width=0.95\linewidth]{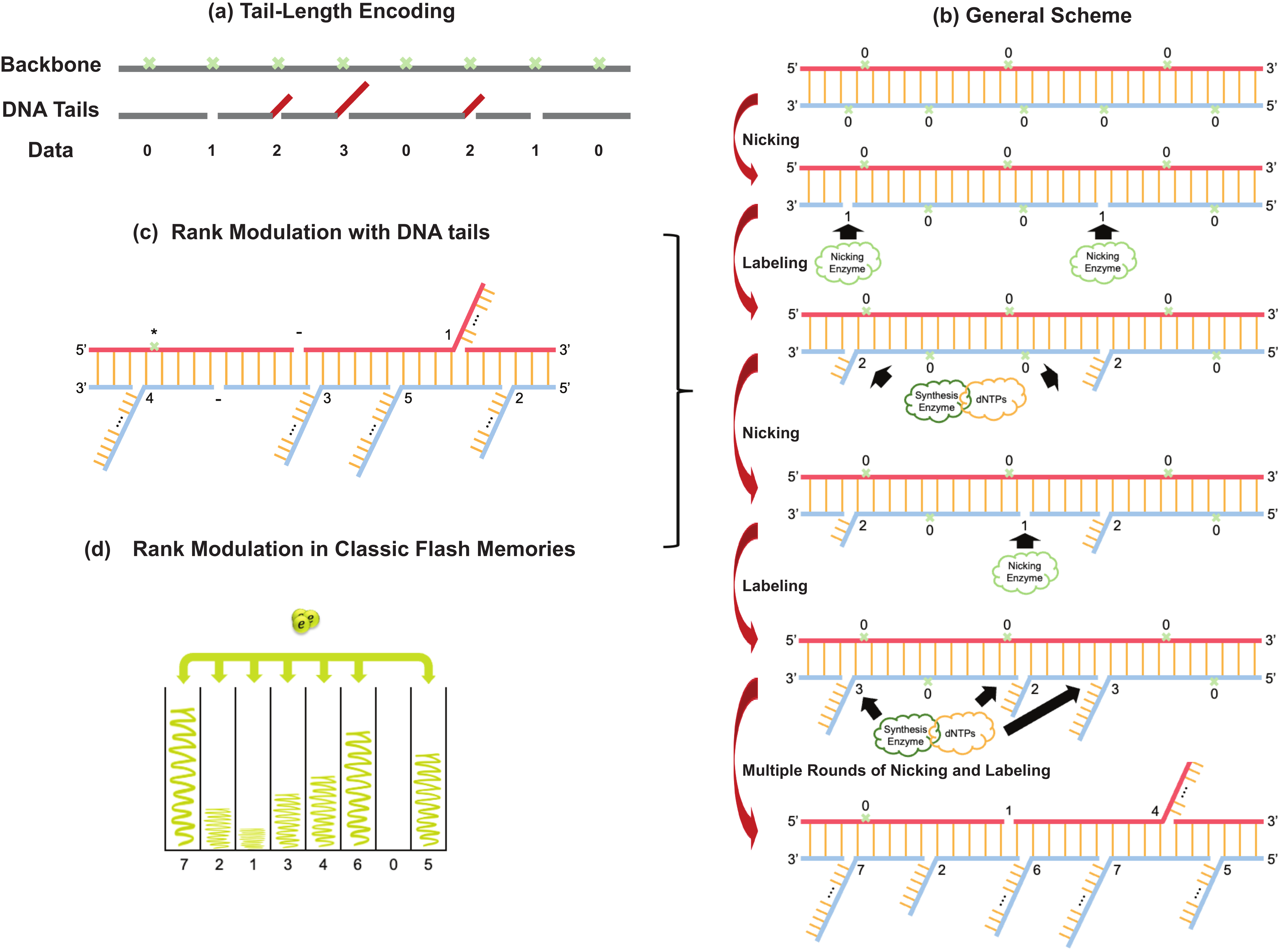}
\vspace{-0.1in}
\caption{An overview of our DNA Tail framework. \textbf{(a)} Schematic of Tail-length encoding, showing the locations were tails can be grown according to the  natural order on the DNA backbone. \textbf{(b)} Schematic of the general multi-round, nonbinary approach for recording information in DNA tails (i.e., single-stranded DNA fragments enzymatically synthesized on double-stranded substrates). For low-cost, we use native restriction endonucleases for nicking and the TdT polymerase for tail growths. \textbf{(c,d)} Schematics of \emph{rank modulation} for tail and cell ``charges.''}
\vspace{-0.25in}
\label{fig:general}
\end{figure}

We used the Tail encoding technique to encode real information in different contexts. Specifically, we illustrate an example of topological tail encoding of metadata equal to the number $20$ on the backbone of synthetic DNA image of Novak Djokovic playing tennis shown in Figure~\ref{fig:rank} (a) to indicate the number of Grand Slam single titles he won until 2021, and the metadata $5030$ into the image of a beach in Uruguay shown in Figure~\ref{fig:rank} (a) to indicate the country's world cup championship years (1930, 1950). We used IDT gBlocks of length $1,000$ bps to record the image content of these two images; metadata is recorded via the general scheme in Figure~\ref{fig:general} (a). The images were first compressed using JPEG, parsed into blocks of length $35$ bits each, and then mapped to DNA sequences of length $19$ nts. The redundant $1.5$ bits per block ensure balanced $GC$ content ($45\%-55\%$) and eliminate homopolymers of length $\geq 3$ nts. To enable random access to different images, we also included pairs of unique prefix and suffix primer sequences for each of the images. Furthermore, to indicate the order of the sequences within the image, we use address blocks of length $3$ nts. We also added $7$ random bases at predefined locations to lower IDT ``synthesis complexity.''

Our experimental results are depicted in Figure~\ref{fig:rank}. In \textbf{(b)}, we show the results of recording a signature DNA tail $20$ on a synthetic DNA image on the right in (a). The value $20$ is nicked into gBlocks by using a combination of two nicking enzymes, Nb.BtsI and Nb.BssSI. To determine how to decode the tail lengths to label values, we performed extensive calibration experiments. The plot summarizes the relationship between the average tail length and the corresponding label for up to $6$ cycles of tail extensions (with $r^2$ denoting the squared fitting error). For an average tail length of $18.16$ as shown in the left plot, the fitted calibration model indicates a corresponding label of $2.46$, marked in red. Since $2.46$ is closer to $2$ than $3$, the label is decoded as $2$. The label $0$ can be perfectly recovered, as it corresponds to the absence of any modifications. In \textbf{(b)}, we provide matching results for gBlocks encoding the right image in (a), with the superimposed value $5030$. Encoding is performed using a combination of four nicking enzymes, Nb.BsmI, Nt.Bpu10I, Nb.BsrDI, and Nb.BssSI. In this case, label $5$ is erroneously read as $6$, while label $3$ is erroneously read as $4$. This further motivates the use of \emph{rank modulation} coding which only requires that the average lengths of the tails be rank-ordered with a sufficiently large difference in values. (d) Rank modulation experiments on the encoding of the poem \emph{A Dream Within a Dream} by E. A. Poe using three gBlocks, with a topologically encoded book ISBN numbers $4570015$ in Poem-GBlock 1, cipher $5010126054$ (Poem-GBlock 2), and $0040721$ (Poem-GBlock 3). The characters of the poem were first converted to binary sequences in ASCII format, parsed into blocks, and mapped to DNA sequences. We note that all rankings of tail lengths (\textbf{(e)}) are consistent with the magnitude of the label, except in Poem-GBlock 1. There, the tail length corresponding to label $7$ ($498.2$) is unacceptably short, falling within the range of lengths designated for label $5$ ($459.07\sim 501.5$). No such inconsistencies are observed in the other two gBlocks. The identified errors suggest that it is possible for some tails to stop growing even in the rank modulation setting, and such errors are studied in the theoretical analysis to follow. 

%%%%%%%%%%%%%%
% rank modulation figure
\begin{figure}
\centering
\includegraphics[width=\linewidth]{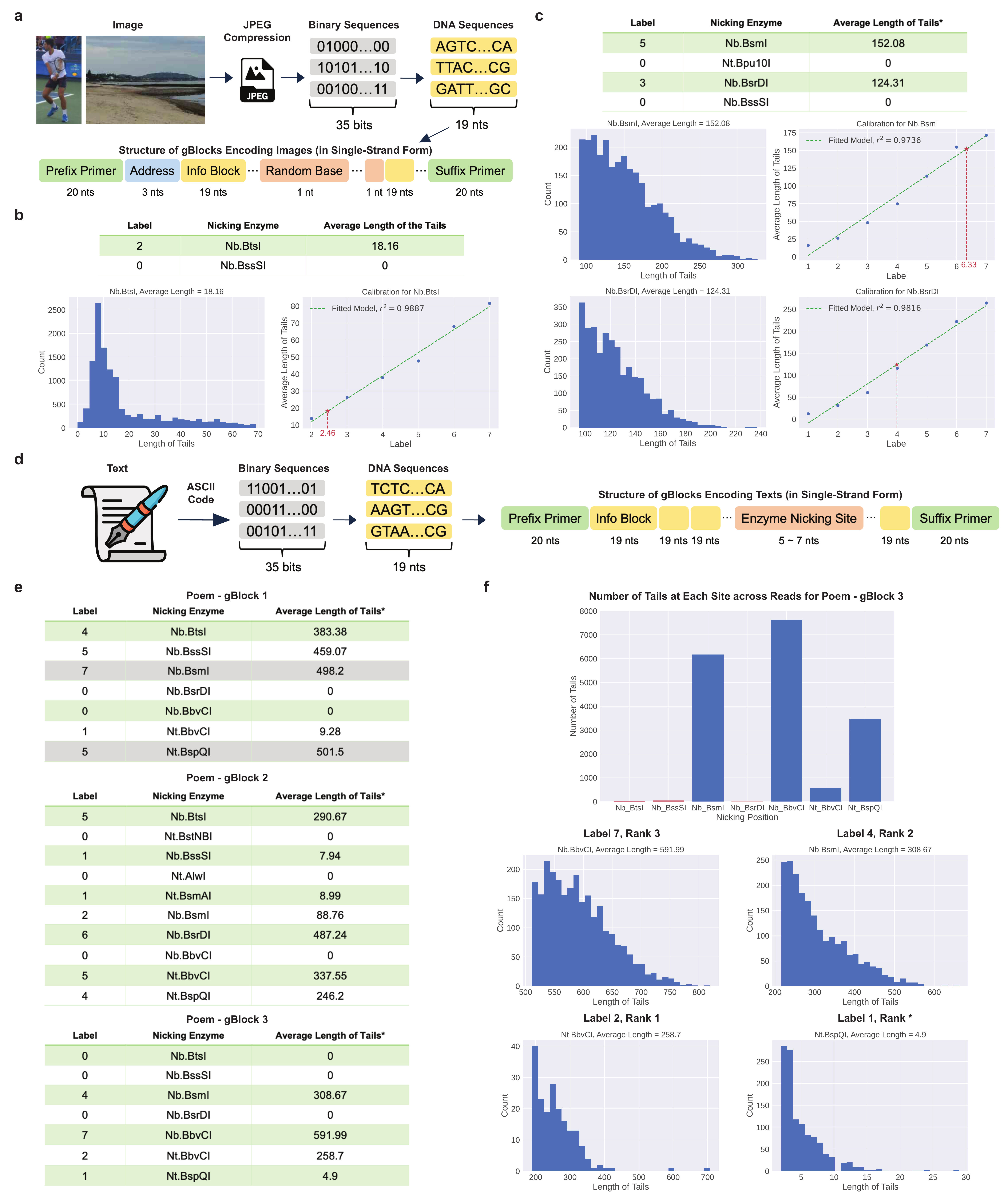}
\caption{\textbf{(a)} Schematic of the image encoding procedure, explained in the main text. 
\textbf{(b)} Results of recording a signature DNA tail $20$ on the first synthetic DNA image. \textbf{(c)} Matching results for gBlocks encoding the right image in (a), with the superimposed value $5030$. (d) Rank modulation experiments on the encoding of the poem \emph{A Dream Within a Dream} by E. A. Poe using three gBlocks. \textbf{(e)} Rank modulation errors.}
\label{fig:rank}
\end{figure}

\clearpage

\section{Error Models for DNA Tails}\label{sec:errormodels}

As evidenced by the experimental results, during tail extension, long tails may experience stumped growth. Moreover, the tail length are random. Therefore, the measured averaged lengths have to be quantized. As a result, the quantized length of a tail corresponding to a larger label can be indistinguishable from that of a tail corresponding to a smaller symbol. These issues introduce new models for rank modulation errors, as described below.

Assume that the DNA tail lengths are encoded via permutations $\sigma = (\sigma(1),\ldots,\sigma(n))\in\mathcal{S}_n$ of length $n$; % such that $\sigma(i)\in [n]=\{1,\ldots,n\}$, $i\in[n]$ and that $\sigma(i)$, $i\in[n]$ are different. 
here, $\mathcal{S}_n$ denotes the set of all permutations, i.e., the symmetric group of order $n!$. The value of a symbol in the permutation represents the quantized tail length at the corresponding nicking site.  For example, the permutation $\sigma=(1,5,2,4,3)$ may represent tail lengths at five nicking sites where the first tail has the shortest length (i.e., length falling in the first quantization bin), and the second has the longest length (i.e., length falling in the last quantization bin). Now, the tail at the fifth nicking site may have stopped properly growing starting from the fourth round or nicking, which could have resulted in it being quantized to $2$, so that $\sigma(2)=2$. That would lead to an erroneous readout $\sigma_e=(1,5,2,4,2)$ from the quantized tail length measurements. The resulting $\sigma_e$ is no longer a permutation due to quantization of average tail lengths, but rather what we refer to as a \emph{multiset permutation} in the sense that it can have repeated or missing values. Also, note that we know that at least one of the two $2$ symbols had to be correct, which provides additional information that can be exploited in the code design process. We hence present three new error models that capture how tail extension and quantization processes affect the permutation received at the decoder. %and propose error-correcting codes for these error models that have the orderwise-optimal redundancy. %It is assumed that the information is  

\textbf{Tails stuck at a quantized length shorter by $1$.}  This model pertains to the case that some tails did not grow in at most one round of extension. Hence, a tail that corresponds to the label $\sigma(i)$ may have an average length that is indistinguishable from that of a tail that corresponds to the label $\sigma(i)-1$. In addition, the tail growth saturation phenomena may arise only for long tails. In this case, the stuck-at errors only occur when $\sigma(i)$ is greater than a threshold $m$. More specifically, let $t$ be the total number of stuck-at errors. Let  $\sigma\in\mathcal{S}_n$ be the permutation encoding user data and let $\sigma_e\in [n]^n,$ where $[n]=\{1,\ldots,n\}$ for any positive integer $n$, be a sequence of quantized tail lengths identified after the average tail quantization processes. A stuck-at error occurs when $\sigma_e(i)=\sigma(i)-1$ for some $i$ such that $\sigma(i)>m$. Hence, the resulting permutation satisfies
\begin{align}\label{eq:kstuckerror}
    \sigma_e(i)=\begin{cases}
        \sigma(i)-1, &\mbox{for $i\in\{i_1,\ldots,i_t\}$ such that $\sigma(i)>m$,}\\
        \sigma(i),&\mbox{for $i\in [n]\backslash\{i_1,\ldots,i_t\}$.}
    \end{cases}
\end{align}
The following is an example of such errors. 
\begin{example}\label{ex:1}
    Let $n=9,t = 3, m=3, \sigma=(9,1,4,2,5,8,3,6,7)$, and $\sigma_e=(8,1,4,2,4,8,3,6,$ $6)$. Then stuck-at errors occurred at nicking sites $1,5$ and $9$, impacting $\sigma(1),\sigma(5)$, and $\sigma(9)$. 
\end{example}
While the stuck-at errors described by~\eqref{eq:kstuckerror} can be considered as $2t$ erasure errors in $\sigma$, we note that these $t$ stuck-at errors are easier to correct than $2t$ general erasure errors since stuck-at errors occur in a permutation sequence and affect only symbols with adjacent values. We will show that the redundancy needed to correct $t$ stuck-at errors is less than that needed to correct $2t$ erasures. Note that a related type of errors is the stuck-at error in write-once memories~\cite{kuznetsov1974coding,wachter2015codes}, where symbols get stuck at a fixed value, but the codewords are not necessarily permutations. In the models considered in this paper, the symbols can be stuck at different values and the codewords are restricted to be permutations.

\textbf{Tails of consecutive lengths stuck at the same length.}
In this model, tails corresponding to consecutive symbol values may stop growing after reaching a certain round of extension. As a result, the average lengths of the corresponding tails are quantized to the lowest observed tail-length value. For example, when encoding $\sigma=(1,6,5,2,4,3)$, the tails at the third and fifth nicking site may have stop growing after they reached the quantized length of bin $3$. Then, the resulting multiset permutation becomes $\sigma_e=(1,6,3,2,3,3)$. 
We say a burst of stuck-at errors of length at most $t$ occur in $\sigma$ if the resulting permutation  $\sigma_e(i)=j$ for all $i$ such that $\sigma(i)\in\{j,j+1,\ldots,j+t_1-1\}$ for some $j\in[n]$ and $t_1\in[t]$, i.e.,
\begin{align}\label{eq:consecutivestuck}
    \sigma_e(i)=\begin{cases}
        j, &\mbox{for $i\in\{i_1,\ldots,i_{t_1}\}$, such that $\sigma(i_\ell)>m$ and $\sigma(i_\ell)=j+\ell-1$, $\ell\in[t_1]$,}\\
        &\mbox{$t_1\in[t]$},\\
        \sigma(i),&\mbox{for $i\in [n]\backslash\{i_1,\ldots,i_{t_1}\}$}.
    \end{cases}
\end{align}
The following is an example of a burst of stuck-at errors. %Note that similar to the previous error model, the resulting $\sigma'$ is a degenerated multiset permutation.
\begin{example}\label{ex:2}
    Let $n=15,t = 3, m=4, \sigma=(9,1,4,2,5,14,10,3,6,13,11,7,12,8,15)$, and $\sigma_e=(8,1,4,2,5,14,8,3,6,13,11,7,12,8,15)$. Then the burst stuck-at error occurs at $\sigma(1)$, $\sigma(7)$, and $\sigma(14)$. 
\end{example}
While the errors described in~\eqref{eq:consecutivestuck} may be viewed as burst erasure errors of length $t$ in $\sigma^{-1}$, we subsequently show that the redundancy needed for correcting stuck-at errors is smaller compared to that of erasures since the former arise in permutations.

\textbf{Tails stuck at a quantized lengths shorter by at most $t$, with tail length rank orderings.} Since the tail length growth is hard to control, it is often hard to recover the label of a tail by measuring its length and quantizing it. Instead, it may be more informative to identify the label of a tail through direct rankings of average tail-lengths. In this case, the labels of multiple (as many as $n-t-m$) tails change as a result of a single tail stuck at a lower length. We consider a single tail length stuck-at error, where a symbol $\sigma(i)>m$ gets stuck at a value $\sigma_e(i)=\sigma(i)-t_1$ for $t_1\in[t]$. The values of the symbols $\sigma(j)$, $\sigma(j)\in [\sigma(i)-1]$ stay the same. In addition, since only relative ranking of quantized length are observed, all symbols with value at least $\sigma(i)+1$ decrease by $1$. Therefore,
\begin{align}\label{eq:rankmod}
    \sigma_e(i)=\begin{cases}
        \sigma(i)-t_1, &\mbox{for some $i=i_1\in[n]$, such that $\sigma(i_1)>m$},\\
        \sigma(i)-1,&\mbox{for $i\in [n]$ such that $\sigma(i)>\sigma(i_1)$},\\
        \sigma(i),&\mbox{else}.
    \end{cases}
\end{align}
\begin{example}\label{ex:3}
    Let $n=9,t = 3, m=3, \sigma=(9,1,4,2,5,8,3,6,7)$, and $\sigma_e=(8,1,4,2,2,7,3,5,$ $6)$. The error that occurs at $\sigma(5)$ results in changes of values of the symbols $\sigma(1),\sigma(5),\sigma(6),\sigma(8),$ and $\sigma(9)$.    
\end{example}
The errors described in~\eqref{eq:rankmod} are related to translocation errors in the 
Ulam distance for rank modulation. While the stuck-at errors in~\eqref{eq:rankmod} can be corrected using codes in the Ulam metric~\cite{farnoud2013error,hassanzadeh2014multipermutation}, we note that the errors in~\eqref{eq:rankmod} preserve part of the positional information about the errors, which is in contrast with the Ulam metric errors for which no  positional information is available. Hence, it is possible to correct stuck-at errors with less redundancy when compared to correcting translocation errors in the Ulam metric.

\section{Codes for $t$ stuck-at errors}\label{sec:codes}
We provide next code constructions for the error models described in Section \ref{sec:errormodels}.  
\subsection{The $t$ stuck-at error model}\label{sec:kstuckerrors}
We start with the $t$ stuck-at error case described in \eqref{eq:kstuckerror} and illustrate the idea through Example \ref{ex:1}. Let the data be encoded by a permutation $\sigma=(9,1,4,2,5,8,3,6,7)$ of length $n=9$. To protect $\sigma$ from at most $t=3$ stuck-at errors that occur at symbols with values larger than $m=3$, we use Lehmer codes (which will be rigorously defined later) of the same length as $\sigma$. In Lehmer encoding of a permutation $\sigma$, the symbol at position $i$ is given by the number of symbols in $\sigma$ that precede position $i$ and have values greater than $\sigma(i)$. For example, the Lehmer encoding of $\sigma=(9,1,4,2,5,8,3,6,7)$ equals $(0,1,1,2,1,1,4,2,2)$. For error correction purposes, we consider the modulo $2$ reduction of the Lehmer encoding of $\sigma$, given by $(0,1,1,0,1,1,0,0,0)$ for the running example. It will be shown that $t$ stuck-at errors result in at most $t$ substitution errors in the modulo $2$ reduction of Lehmer encodings. %For example, if $\sigma'=(8,1,4,2,4,8,3,6,6)$, the modulo $2$ of the Lehmer\footnote{Normally, Lehmer codes are defined for permutations. Here we relaxed the definition so that they are defined for any vector.} encoding of $\sigma'$ becomes $(0,1,1,0,1,0,0,0,0)$, which differs from that of $\sigma$ by the 6th bit. 
To correct $t$ such substitution errors with known locations in the vector, it suffices to use a $t$-erasure correcting Reed-Solomon code with at most $t\log (n-m)$ redundant bits. %, since the positions of the substitution errors can be deduced. 
In addition, one can recover $\sigma$ from $\sigma_e$ and the modulo $2$ reduction of the Lehmer encoding of $\sigma$. %For example, to recover $\sigma$ from $\sigma'$, we need to decide whether $\sigma(6)=9$ or $\sigma(1)=9$, whether $\sigma(3)=5$ or $\sigma(5)=5$, and whether $\sigma(8)=7$ or $\sigma(9)=7$. To this end, we only need to check the 5th, 6th, and 9th bits where the modulo $2$ of the Lehmer encodings of $\sigma'$ and $\sigma$ differ. Since these two vectors only differ in the 6th bit, we conclude that the symbol $9$ occurs before $8$ in $\sigma$ and the symbols $5$ and $7$ occur after symbols $4$ and $6$, respectively, and hence $\sigma$ can be recovered.  

Since codewords are permutations in our model, one needs to encode the binary Reed-Solomon code redundancy into ``permutation symbols.'' We utilize the fact that only symbols with values larger than $m$ can be affected by errors and assume that $m\ge  \frac{t\log (n-m)}{\log n}+2$, which is typically the case in our experiments. We then use the positional information of the symbols in $[\lceil \frac{t \log (n-m)}{\log n}\rceil]$ to store the redundant symbols. The symbols $[n+\lceil \frac{t \log (n-m)}{\log n}\rceil]\backslash [\lceil \frac{t \log (n-m)}{\log n}\rceil]$ encode the information in $\sigma$, where each symbol $\sigma(i)$ is simply encoded as $\sigma(i)+\lceil \frac{t \log (n-m)}{\log n}\rceil$. %The idea is to use the positional information of the symbols $[\lceil \frac{k \log (n-m)}{\log n}\rceil]$ to encode the Reed-Solomon code redundancy. 
For example, assume that the Reed-Solomon redundancy is given by three $9$-ary symbols, $(1,0,7)$. In this case, we increase each entry in $\sigma$ by $3$ so that $\sigma =(12,4,7,5,8,11,7,9,10)$ and then insert symbols $1,2$, and $3$ after the 1st, 0th (which is before the first), and 7th entry in $\sigma$ to obtain the encoded permutation $(2,12,1,4,7,5,8,11,7,3,9,10)$.

In what follows, we provide more details about the encoding and decoding procedures, and prove the following theorem, which shows that the stuck-at errors can be corrected by adding at most $t$ redundant symbols to the permutation $\sigma$.% For notation convenience, in the following, it is assumed that $n$ is a power of $2$. 
\begin{theorem}\label{thm:kstuckerror}
For any message given in the form of a permutation $\sigma$ of length $n$, there is an encoder mapping $\mathcal{E}:\mathcal{S}_n\rightarrow\mathcal{S}_{n+t'}$ that maps $\sigma$ to a permutation $\mathcal{E}(\sigma)$ of length $n+t'$, where $t'\ge \frac{t\log (n-m)}{\log n}$. Moreover, $\mathcal{E}(\sigma)$ can be corrected from at most $t$ stuck-at symbol errors defined in~\eqref{eq:kstuckerror}, given $m\ge t'+2$.
\end{theorem}
\begin{remark}
There are $\binom{n-m-t-1}{t}$ choices for the locations of $t$ stuck-at errors in~\eqref{eq:kstuckerror}, all resulting in different erroneous permutations. By the sphere packing bound, the redundancy of a stuc-at error-correcting code is at least $\log \binom{n-m-t-1}{t}=O(t\log (n-m))$. 
\end{remark}
Before presenting the code construction, we first give a formal definition of Lehmer codes. 
For any sequence $\pi\in [n]^n$, its Lehmer encoding $\mathcal{L}(\pi)\in\{0\}\times[1]\times[2]\ldots\times[n-1]$ equals
\begin{align}\label{eq:lehmer}
\mathcal{L}(\pi)(i)=|\{j:j<i,\pi(j)>\pi(i)\}|.
\end{align}
Note that $\pi$ %can be a permutation in $\mathcal{S}_n$ and 
is not necessarily a permutation. The following Lemma shows how stuck-at errors in $\sigma$ affect $\mathcal{L}(\sigma)$.
\begin{lemma}\label{lem:recoverfromlehmer}
Let $\sigma_e$ be an erroneous version of $\sigma$ such that
\begin{align}\label{eq:erroneoussigma}
 \sigma_e(i) = \begin{cases}
     \sigma(i)-1,&\textup{for $i\in[n]$ such that $i\in\{i_1,\ldots,i_\ell\}$, $\sigma(i)>m$, and},\\
     &\textup{$\sigma(i_j)\le \sigma(i_{j+1})-2$ for $j\in[\ell-1]$},\\
     \sigma(i),&\textup{for $i\in[n]\backslash\{i_1,\ldots,i_\ell\}$,}
 \end{cases}   
\end{align}
for $\ell\le t$. Moreover, $\sigma_e$ has two repeated symbol values $\sigma_e(i_j)=\sigma_e(i'_j)=\sigma(i_j)-1$ for $j\in[\ell]$. Then,
\begin{align}\label{eq:recoverfromlehmer}
\mathcal{L}(\sigma_e)(i)=\begin{cases}
    \mathcal{L}(\sigma)(i)-1,&\textup{if $i=i'_j$ and $ i'_j>i_j$ for some $j\in[\ell]$,}\\
    \mathcal{L}(\sigma)(i),&\textup{otherwise.}
\end{cases}   
\end{align}
\end{lemma}
\begin{proof}
We show that for any $i,i'\in [n]$ and $i<i'$, we have $\sigma_e(i)>\sigma_e(i'
)$ if and only if $\sigma(i)>\sigma(i'
)$, unless $\sigma_e(i)=\sigma_e(i')$ and $i=i_j=\min\{i_j,i'_j\}$ for some $j\in[\ell]$. Suppose we have either $\sigma_e(i)>\sigma_e(i'
)$ and $\sigma(i)\le \sigma(i'
)$ or $\sigma_e(i)\le\sigma_e(i'
)$ and $\sigma(i)> \sigma(i'
)$. If $\sigma_e(i)>\sigma_e(i'
)$ and $\sigma(i)\le \sigma(i'
)$, then $\sigma(i')-1\ge\sigma(i)\ge \sigma_e(i)> \sigma_e(i'
)\ge \sigma(i')-1,$
which is a contradiction. On the other hand, if $\sigma_e(i)\le\sigma_e(i'
)$ and $\sigma(i)> \sigma(i'
)$, we have $\sigma(i)> \sigma(i'
)\ge \sigma_e(i'
)\ge \sigma_e(i)\ge \sigma(i)-1.$ Hence, $\sigma_e(i)=\sigma_e(i')$, $i=i_j=\min\{i_j,i'_j\}$, and $i'=i'_j$ for some $j\in[\ell]$. Therefore, $\mathcal{L}(\sigma_e)(i')=\mathcal{L}(\sigma)(i')-1$ if and only if $i'=i'_j$ and $i'_j>i_j$ for some $j\in[\ell]$. %Hence, the proof is done.
\end{proof}
% Since a single symbol stuck-at error that occurs at $\sigma(i)$ confuses the two symbols $\sigma(i)$ and $\sigma(i)-1$, we use a single bit
% \begin{align*}
%     \mathbbm{1}(\sigma^{-1}(\sigma(i)-1)>i)=\begin{cases}
%         1 & \text{if } i_1>i \text{ for }\sigma(i_1)=\sigma(i)-1 \\
%         0 & \text{else}
%     \end{cases}
% \end{align*}
% to indicate whether there is a positional inversion between the labels with values $\sigma(i)$ and $\sigma(i)-1$ for every $i$ such that $\sigma(i)>m$. Define the indicator vector
% \begin{align*}
%     I(\sigma) = (\mathbbm{1}(\sigma^{-1}(m)>\sigma^{-1}(m+1)),\ldots,\mathcal{1}(\sigma^{-1}(n-1)>\sigma^{-1}(n))).
% \end{align*}
The following lemma shows that for any $\sigma_e$ satisfying~\eqref{eq:kstuckerror}, we can give an estimate $\hat{\sigma}$ of $\sigma$ based on $\sigma_e$ that satisfies \eqref{eq:erroneoussigma}. 
\begin{lemma}\label{lem:kstuckerror}
For any $\sigma_e$ be given by~\eqref{eq:kstuckerror}, one can obtain an estimate $\hat{\sigma}$ of $\sigma$ that satisfies~\eqref{eq:erroneoussigma}.
\end{lemma}
\begin{proof}
Let $\sigma_e$ be obtained from $\sigma$ after stuck-at errors at symbols whose values belong to the union of disjoint intervals $\cup^{L}_{\ell=1}\{i'_\ell+1,\ldots,i'_{\ell}+j_\ell\}$ such that $\sum^{L}_{\ell=1}j_\ell\le t$ and that $i'_{\ell}+j_\ell+1<i'_{\ell+1}$. Then, for each $\ell\in[L]$, there are two symbols with repeated values $i'_\ell$ in $\sigma_e$, one of which comes from the symbol in $\sigma$ with value $i'_\ell+1$. Moreover, the symbols with values $i'_\ell+1,\ldots,i'_{\ell}+j_\ell-1$ in $\sigma_e$ arise from symbols in $\sigma$ with values $i'_\ell+2,\ldots,i'_{\ell}+j_\ell$, respectively. The symbol with value $i'_{\ell}+j_\ell$ does not appear in $\sigma_e$. 

To obtain $\hat{\sigma}$ from $\sigma_e$, we find the missing values in $\sigma_e$, which coincide with the values $i'_{\ell}+j_\ell$ for $\ell\in[L]$. Then, for each missing value $i'_{\ell}+j_\ell$ we find the largest repeated value in $\sigma_e$ that is smaller than $i'_{\ell}+j_\ell$, and this coincides with $i'_\ell$. 
Let 
\begin{align*}
    \hat{\sigma}(i) = \begin{cases}
        \sigma_e(i)+1,&\text{if }\sigma_e(i)\in\cup^L_{\ell=1}\{i'_\ell+1,\ldots,i'_{\ell}+j_\ell-1\},\\
        \sigma_e(i),&\text{else}.
    \end{cases}.
\end{align*}
Note that the values $i'_\ell$ and $j_\ell$, $\ell\in[L]$ can be inferred from $\sigma_e$ as described above. 
Then, 
\begin{align}
    \hat{\sigma}(i) = \begin{cases}
        \sigma(i)-1,&\text{if }\sigma(i) \in \cup^L_{\ell=1}\{i'_\ell+1\},\\
        \sigma(i),&\text{else}
    \end{cases}.
\end{align}
Moreover, we have that $i'_\ell+2\le i'_{\ell+1}$ by definition of $i'_\ell$. Hence $\hat{\sigma}$ satisifies \eqref{eq:erroneoussigma}.
%Hence, $I(\hat{\sigma})$ and $I(\sigma)$ differ in at most $L\le k$ bits $\mathbbm{1}(\sigma^{-1}(i'_\ell)>\sigma^{-1}(i'_\ell+1))$, $\ell\in[L]$. Note that the positions $i'_\ell$ are known.
%Finally, the bit $\mathbbm{1}(\sigma^{-1}(i'_\ell)>\sigma^{-1}(i'_\ell+1))$ can be used to distinguish between the symbols with repeated values $i'_\ell$ in $\hat{\sigma}$ and infer which of the symbol with value $i'_\ell$ in $\hat{\sigma}$ comes from the symbol with value $i'_\ell+1$ in $\sigma$. Therefore, $\sigma$ can be recovered from $I(\sigma)$ and $\hat{\sigma}$.
\end{proof}
According to Lemma~\ref{lem:kstuckerror}, one can reduce the problem of recovering $\sigma$ from $\sigma_e$ satisfying~\eqref{eq:kstuckerror} to that of recovering $\sigma$ from $\sigma_e$ satisfying~\eqref{eq:erroneoussigma}. Furthermore, 
based on Lemma~\ref{lem:recoverfromlehmer}, we will consider the modulo $2$ reduction of $\mathcal{L}(\sigma)$, and only focus on symbols with values larger than $m$, i.e., 
$$\mathcal{B}(\sigma)=(\mathcal{L}(\sigma)(i)\bmod 2:\sigma(i)>m),$$
for $i\in [n]$. Lemma~\ref{lem:recoverfromlehmer} shows when $\sigma_e$ satisfies~\eqref{eq:erroneoussigma}, $\mathcal{B}(\sigma_e)$ changes in at most $t$ positions $i$, where $i=i'_j$ and $i'
_j>i_j$ for some $j\in[\ell]$. Hence, $t$ stuck-at errors result in at most $t$ substitutions in $\mathcal{B}(\sigma)$, the positions of which can be inferred. Moreover, no errors occur in $\mathcal{L}(\sigma)(i)$ for $\sigma(i)\le m$.

To protect $\mathcal{B}(\sigma)$ from $t$ erasures, we use Reed-Solomon codes. Specifically, we encode a binary sequence $\boldsymbol{x}\in\{0,1\}^\ell$ of length $\ell$ into a sequence over an alphabet of size $q$ by first splitting $\boldsymbol{x}$ into blocks $\boldsymbol{x}_i$, $i\in[\frac{\ell}{\log q}],$ of length $\log q,$ where each block is represents by a symbol from the alphabet of size $q$ of the Reed-Solomon code. Let $RS_t(\boldsymbol{x}):\{0,1\}^{\ell}\rightarrow [q]^t$ be a mapping such that $(\boldsymbol{x}_1,\ldots,\boldsymbol{x}_{\frac{\ell}{\log q}},RS_t(\boldsymbol{x}))$ is a Reed-Solomon code capable of correcting $t$ symbol erasures. It is required that $q\ge t+\frac{\ell}{\log q}+1$. We let $q=n$ and $\ell= n-m$. Note that $q\ge t+\frac{\ell}{\log q}+1$ is satisfied when $n>4$ and $t<n$. 

As mentioned in the illustrating example, one needs to encode $RS_t(\mathcal{B}(\sigma))$ in permutations. To this end, we use the fact that permutations of length $n$ are over the  alphabet $[n]$ and use redundant symbols to encode $RS_t(\mathcal{B}(\sigma))$. %Moreover, since the redundant symbols are subject to stuck-at errors, where each stuck-at error makes two symbols with consecutive values indistinguishable, one might need to protect number of symbols that is twice the number of stuck-at errors. In our encoding, we avoid this by using symbols with nonconsecutive values to encode the Reed-Solomon code redundancy. 
We use the symbols with values in $[t']$ to encode $RS_t(\mathcal{B}(\sigma))$. Note that under the assumption $m\ge t'+2$, the symbols with values in $[t']$ can still be identified/recognized after $t$ stuck-at errors. 
Moreover, we encode the Reed-Solomon redundancy $RS_t(\mathcal{B}(\sigma))$ using positional information rather than the actual values of the redundant symbols. As a result, the original permutation $\sigma$ is encoded using symbols with values in  $[n+t']\backslash[t']$. The details of the encoding  %decoding 
procedure are as follows.  

\textbf{Encoding:}
\begin{itemize}
    \item[\textbf{(1)}] Given a permutation $\sigma\in\mathcal{S}_n$, compute the redundancy $RS_t(\mathcal{B}(\sigma))$ and represent it by $t'$ symbols $(r_1,\ldots,r_{t'})$ over the alphabet $[n]$. %which has a one-to-one mapping to the alphabet $[n]$. 
    \item[\textbf{(2)}] Compute $\mathcal{F}(\sigma)$ by $\mathcal{F}(\sigma)(i)= \sigma(i)+t'$ for $i\in[n]$.  
    \item[\textbf{(3)}] Insert $i\in[t']$, right after the $r_i$th symbol $\sigma(r_i)$ in $\sigma$. If $r_i=r_j$ for $i<j\in [t']$, insert $j$ after $i$ where $i$ and $j$ are between the $r_i$th symbol and the $r_i+1$th symbol in $\mathcal{F}(\sigma)$.
\end{itemize}
Let $\mathcal{E}(\sigma)\in\mathcal{S}_{n+t'}$ be the output of the encoding algorithm. Note that $\sigma$ is encoded in the symbols of values $[n+t']\backslash[t']$ in $\mathcal{E}(\sigma)$. 
The decoding procedure works as follows.%is based on Lemma \ref{lem:kstuckerror}.

\textbf{Decoding:}
\begin{itemize}
    \item[\textbf{(1)}] Given an erroneous permutation of $\mathcal{E}(\sigma)$, compute an estimate $\hat{\mathcal{E}}(\sigma)$ of $\mathcal{E}(\sigma)$ according to~Lemma~\ref{lem:kstuckerror}. 
    \item[\textbf{(2)}] 
    Let $r_i=|\{j:j<\ell,\hat{\mathcal{E}}(j)\in [n+t']\backslash[t'],\hat{\mathcal{E}}(\ell)=i\}|$ be the number of symbols in $\hat{\mathcal{E}}$ that precede the symbol $i$ and have values in $[n+t']\backslash[t']$. 
    \item[\textbf{(3)}] Let $\hat{\mathcal{F}}(\sigma)$ be an estimate of $\mathcal{F}(\sigma)$ obtained from $\hat{\mathcal{E}}$ by removing symbols with values in $[t']$ and subtracting $t'$ from each entry. 
    Compute $\mathcal{B}(\hat{\mathcal{F}}(\sigma))$ and determine the erasure positions based on Lemma~\ref{lem:recoverfromlehmer}. Then use $(r_1,\ldots,r_{t'})$ as Reed-Solomon redundancy to correct erasures in $\mathcal{B}(\hat{\mathcal{F}}(\sigma))$ and obtain $\mathcal{B}(\sigma)$.
    \item[\textbf{(4)}] Recover $\sigma$ from $\hat{\mathcal{F}}(\sigma)$, $\mathcal{B}(\hat{\mathcal{F}}(\sigma))$, and  $\mathcal{B}(\sigma)$, based on Lemma~\ref{lem:recoverfromlehmer} as follows. Let $\hat{\mathcal{F}}(\sigma)(i'_j)=\hat{\mathcal{F}}(\sigma)(i_j)$, $j\in[\ell],$ be the $\ell$ pairs of repeated symbols in $\hat{\mathcal{F}}(\sigma)$. For each $j\in[\ell]$, if $\mathcal{B}(\hat{\mathcal{F}}(\sigma))(i_j)=\mathcal{B}(\sigma)(i_j)$ and $\mathcal{B}(\hat{\mathcal{F}}(\sigma))(i'_j)=\mathcal{B}(\sigma)(i'_j)$, then let $\hat{\mathcal{F}}(\sigma)(\min\{i_j,i'_j\})=\hat{\mathcal{F}}(\sigma)(i_j)+1$. Otherwise, let $\hat{\mathcal{F}}(\sigma)(\max\{i_j,i'_j\})=\hat{\mathcal{F}}(\sigma)(i_j)+1$.
    \item[\textbf{(5)}] 
    Output $\hat{\mathcal{F}}(\sigma)$, the estimate of $\sigma$.
    %Then, according to Lemma \ref{lem:boundonerasurekstuckerror}, $(\boldsymbol{X},r_1,\ldots,r_{k+2\frac{k}{\log n}})$ is an erroneous version of $(\boldsymbol{I}(\sigma),RS_{k+2\frac{k}{\log n}}(I(\sigma)))$ with up to $k+2\frac{k}{\log n}$ erasures. 
\end{itemize} 
We next prove the correctness of the decoding procedure. Note that by assumption, $m\ge t'+2$ and hence the symbols $1,\ldots,t'$ are not affected by errors and hence $(r_1,\ldots,r_{t'})=RS_t(\mathcal{B}(\sigma))$ is correctly decoded. Moreover, $\hat{\mathcal{F}}(\sigma)$ is an erroneous version of $\sigma$ satisfying~\eqref{eq:erroneoussigma}. Hence, by Lemma~\ref{lem:recoverfromlehmer}, $\mathcal{B}(\hat{\mathcal{F}(\sigma)})$ differs from $\mathcal{B}(\sigma)$ in at most $t$ bits, the positions of which can be determined. Then, $\mathcal{B}(\sigma)$ can be recovered with the help of the Reed-Solomon code redundancy $(r_1,\ldots,r_{t'})$. According to Lemma~\ref{lem:recoverfromlehmer}, for each $i\in[n]$ where $\mathcal{B}(\hat{\mathcal{F}(\sigma)})(i)$ and $\mathcal{B}(\sigma)(i)$ differ, we have $\mathcal{L}(\hat{\mathcal{F}}(\sigma))(i)=\mathcal{L}(\sigma)(i)-1$. For other values of $i$ we have $\mathcal{L}(\hat{\mathcal{F}}(\sigma))(i)=\mathcal{L}(\sigma)(i)$. Hence, according to Lemma~\ref{lem:recoverfromlehmer}, the estimate $\hat{\mathcal{F}}(\sigma)$ in Step (4) of decoding equals $\sigma$. %Therefore, $\sigma$ is correctly decoded in Step (5) of decoding.  
\subsection{The burst stuck-at error model}\label{sec:burststuckerror}
We now provide code constructions for cases when symbols with at most $t$ consecutive values get stuck, which is described by~\eqref{eq:consecutivestuck}. Suppose data is encoded into a permutation $\sigma=(9,1,4,2,5,14,10,3,6,13,11,$ $7,12,8,15)$ of length $15$ and at most $t=2$ stuck-at errors occur at symbols with values larger than $m=3$. We group symbol values $\{1,\ldots,15\}$ into blocks of length $2t=4$, i.e., $\{1,2,3,4\},\{5,6,7,8\},\{9,10,11,12\}$, and $\{13,14,15\}$ (the last block may have fewer than $2t=4$ symbols). For each block of values $(j,j+1,j+2,j+3)$, we look at the relative positions of symbols with these values in $\sigma$ and obtain a permutation $\sigma_{j}$ of length $4$ such that $\sigma^{(-1)}_{j}(i_1)>\sigma^{(-1)}_{j}(i_2)$ if $\sigma^{-1}(j+i_1-1)>\sigma^{-1}(j+i_2-1)$. For block $\{1,2,3,4\}$, the relative ranking is given by $(1,4,2,3),$ since this is the order of symbols $1,2,3$, and $4$ in $\sigma$. Similarly, the blocks $\{5,6,7,8\},\{9,10,11,12\}$ and $\{13,14,15\}$ result in the relative rankings $(1,2,3,4),(1,2,3,4)$ and $\{2,1,3\}$, respectively. In addition to the blocks obtained by grouping values in $[15]$, we create another set of blocks that shifts the values of the first set of blocks by $t$. More specifically, we group $\{1+t=3,\ldots,15\}$ into another set of blocks of length $2t=4$, and compute the relative ranking of the blocks as $\{3,4,5,6\},\{7,8,9,10\},\{11,12,$ $13,14\},$ and $\{15\}$ and obtain $(2,3,1,4),(3,4,1,2),(4,3,1,2)$, and $(1)$, respectively.  
Note that $t=2$ stuck-at errors obfuscate exactly one block in at least one of the two sets of blocks, the identity of which can be determined. %corresponding relative rankings, the index of which can be identified. 
Hence, it suffices to protect from a single erasure of the relative ranking of a single block in both sets of blocks. To this end, we compute the symbol-wise sum of block relative rankings in both sets of blocks, respectively, modulo $2t=4$, while padding with zeros all rankings shorter than $4$. %which is 
%$ (1,1,3,3) $ for the first set of block rankings and $(2,2,3,3)$ for the second set of block rankings. 
%It can be shown that the permutation $\sigma$ can be recovered from $\sigma'$ and the two modulo sums $(1,1,3,3)$ and $(2,2,3,3)$. 
Then, it remains to encode the modulo sums into a permutation $\sigma$. 

Similar to Section \ref{sec:kstuckerrors}, %our idea is to view the concatination of $(1,1,3,3)$ and $(2,2,3,3)$ as the $4$-ary representation of an integer $24495\in \{0,\ldots,4^8-1\}$ and  represent the integer as $5$ $n=15$-ary symbols $(0,7,3,13,0)$ of length $5$. Then, encode $(0,7,3,13,0)$ 
we use the positional information of redundant symbols for encoding. Different from Section~\ref{sec:kstuckerrors}, where it is assumed that the redundant symbols are at most $m$ and do not suffer from errors, here we consider the case when $m$ can be small such that redundant symbols also suffer from stuck-at errors.% where each stuck-at error might affect $k=2$ symbols. 
To avoid a stuck-at error affecting multiple redundant symbols, we interleave the values of symbols that encode $\sigma$ and the values of the redundant symbols such that we use the values $6,9,12,15,18$, and $21$ with difference $t+1=3$ for redundant symbols and encode $\sigma$ in the remaining values $\{1,2,3,4,5,7,8,10,11,13,14,16,17,19,20\}$, for the case of our running example. %while keeping their relative orders in $\sigma$. 
%In this manner, the permutation $\sigma$ in Example \ref{ex:2} is encoded by $\sigma = (11,1,4,2,5,19,13,3,7,17,14,8,16,10,20)$, where each entry is the $\sigma(i)$th smallest value in the above remaining values. 
Moreover, we use an extra redundant symbol to protect the symbols that encode redundancy.%symbols that encode the redundancy $(0,7,3,13,0)$.% This is because we add an extra symbol $7$ to $(0,7,3,13,0)$ such that the sum of symbols $(0,7,3,13,0,7)$ is $0$ modulo $15$. 
%The final  similar to Section \ref{sec:kstuckerrors}, the redundant symbols $6,9,12,15,18$, and $21$ are inserted after the 0th (before the 1st), 7th, 3rd, 13th, 0th, and 7th symbol, respectively, in the encoded $\sigma$ and thus the encoded permutation becomes $(6,18,11,1,4,12,2,5,19,13,9,21,3,7,17,14,8,16,15,10,20)$.

The details are given in the proof of the following theorem, which shows that it suffices to use at most $\frac{4t\log t}{\log n}+1$ redundant symbols to correct a burst of at most $t$ stuck-at errors.
\begin{theorem}\label{thm:burststuckerror}
For any message given in the form of a permutation $\sigma$ of length $n\ge 2t(t+1)$, there is an encoding mapping $\mathcal{E}_b:\mathcal{S}_n\rightarrow\mathcal{S}_{n+t'+1}$ that maps $\sigma$ to a permutation $\mathcal{E}_b(\sigma)$ with length $n+t'+1$ such that $t'\log n \ge 4t\log t$. Moreover, $\mathcal{E}_b(\sigma)$ can be corrected from at most $t$ stuck-at symbol errors described in~\eqref{eq:consecutivestuck}.    
\end{theorem}
\begin{remark}
Note that the amount of information needed to distinguish different relative orderings of the stuck symbols is at least $\log t!=O(t\log t)$. Hence, the redundancy of the code is at least $O(t\log t).$
\end{remark}
Before presenting the code construction, we first introduce the notion of projection of a permutation. For a permutation $\sigma$ and a subset of positions $A=\{i_1,\ldots,i_{|A|}\}\subseteq [n]$, $\sigma_A\in \mathcal{S}_{|A|}$ is a permutation of length $|A|$ such that $\sigma_A(j_1)<\sigma_A(j_2)$ if $\sigma(i_{j_1})<\sigma(i_{j_2})$ for $j_1,j_2\in [|A|]$, i.e., $\sigma_A$ is the relative ranking of symbols in $\sigma$ with positions in $A$. For each $i\in[\lceil\frac{n}{2t}\rceil]$, let
\begin{equation}\label{eq:sigmai1i2}
\sigma^{i,1}=\sigma_{\{\sigma^{-1}(2(i-1)t+1),\ldots,\sigma^{-1}(2it)\}}\in\mathcal{S}_t, \;\;\;
\sigma^{i,2}=\sigma_{\{\sigma^{-1}(t+2(i-1)t+1),\ldots,\sigma^{-1}(t+2it)\}}\in\mathcal{S}_t,
\end{equation}
such that $\sigma^{i,1}(j)=0$ when $2(i-1)t+j$ is not in $\sigma$ and $\sigma^{i,2}(j)=0$ when $t+2(i-1)t+j$ is not in $\sigma$.
Consider the following two concatenations of $\sigma^{i,1}$ and $\sigma^{i,2}$, respectively,
\begin{equation}\label{eq:s1s2}
    S_1 = (\sigma^{1,1},\ldots,\sigma^{\lceil\frac{n}{2t}\rceil,1}), \;\;\;
    S_2 = (\sigma^{1,2},\ldots,\sigma^{\lceil\frac{n-t}{2t}\rceil,2}).
\end{equation}
Note that both $S_1$ and $S_2$ are obtained by splitting the values of symbols in $\sigma$ into blocks of length $2t$ and concatenating the projection of $\sigma$ onto the symbols with these blocks of values. Moreover, there is a $t$-symbol shift between the sets of blocks that are used to construct $S_1$ and $S_2$, respectively. The following lemma shows that either $S_1$ or $S_2$ can be identified to have a single block permutation projection erasure in one of $\sigma^{1,1},\ldots,\sigma^{\lceil\frac{n}{2t}\rceil,1}$ or $\sigma^{1,2},\ldots,\sigma^{\lceil\frac{n-t    }{2t}\rceil,2}$, respectively, under the burst stuck-at error model of~\eqref{eq:consecutivestuck}.
\begin{lemma}\label{lem:consecutivestuck}
Declare an erasure of $\sigma^{i,1}$ or $\sigma^{i,2}$ if at least one value among $2(i-1)t+1,\ldots,2it$ or $t+2(i-1)t+1,\ldots,t+2it$ is missing in $\sigma_e$, respectively, where $\sigma_e$ is as described in~\eqref{eq:consecutivestuck}. Then, at least one of $S_1$ or $S_2$ has at most one declared erasure.   
\end{lemma}
\begin{proof}
Let $j_1$ be the smallest symbol value that got stuck. If $(2i-1)t+1\le j_1\le 2it$ for some $i\in [\lceil\frac{n}{2t}\rceil]$, then only a single erasure of $\sigma^{i,1}$ is declared in $S_1$. On the other hand, if $t+(2i-1)t+1\le j_1\le t+2it$ for some $i\in [\lceil\frac{n-t}{2t}\rceil]$, then only a single erasure of $\sigma^{i,1}$ is declared in $S_1$. Note that the values of the stuck-at symbols can be inferred from $\sigma_e$. 
\end{proof}
According to Lemma~\ref{lem:consecutivestuck}, it suffices to add redundant symbols to protect one permutation projection erasure in $S_1$ and $S_2$, respectively, to correct a burst stuck-at error of length at most $t$. This can be done by representing each permutation projection $\sigma^{i,1}$ or $\sigma^{i,2}$ via a vector of $t$ symbols over an alphabet of size $t$. Then, we use 
\begin{equation}\label{eq:r1r2}
R_1=\oplus_{i\in [\lceil\frac{n}{2t}\rceil]}\sigma^{i,1}, \;\;\;
p_2 =\oplus_{i\in [\lceil\frac{n-t}{2t}\rceil]}\sigma^{i,2}    
\end{equation}
 to protect $S_1$ and $S_2$ from a single erasure, respectively, where $\oplus$ denotes the symbol-wise addition of $\sigma^{i,1}$ or $\sigma^{i,2}$ modulo $t$. 
%The manner in which redundancy $R_1$ and $R_2$ is added to $\sigma$ is similar to the encoding/decoding procedures in Section \ref{sec:stuckerror}. 
Let the concatenation of $R_1$ and $R_2$ be the $t$-ary representation of an integer in the set $\{0,\ldots,t^{4t}-1\}$ and represent the integer by $t'= \frac{4t\log t}{\log n}$ symbols $(r_1,\ldots,r_{t'})$ over an alphabet of size $n$. We encode $\sigma$ and the redundant symbols $(r_1,\ldots,r_{t'})$ that represent $R_1$ and $R_2$ using $n+t'+1$ symbols in total, where symbols with values $n+t'+1-(t+1)(t'+1)+(t+1)i$, $i\in[t']$ are used to encode $(r_1,\ldots,r_{t'})$. We then use the symbol with value $n+t'+1$ to encode an $n$-ary symbol $\sum^{t'}_{i=1}r_i\bmod n$, which represents the redundancy to protect $(r_1,\ldots,r_{t'})$ from a single erasure. The remaining $n$ symbols in the set $V=[n+t'+1]\backslash (\cup_{i\in\{0,\ldots,t'\}}\{n+t'+1-i(t+1)\})$ are used to encode $\sigma$, where $\sigma(i)$ is replaced by the $i$th smallest value in $V$.

\textbf{Encoding:}
\begin{itemize}
    \item[\textbf{(1)}] Given a permutation $\sigma\in\mathcal{S}_n$, use the symbols of values in $V=[n+t'+1]\backslash (\cup_{i\in\{0,\ldots,t'\}}\{n+t'+1-i(t+1)\})$ to encode $\sigma$. More specifically, let $\mathcal{F}(\sigma)(i)$ be the $\sigma(i)$th smallest value in $V$, $i\in[n]$.  
    \item[\textbf{(2)}] Find the sequences $S_1$ and $S_2$ according to~\eqref{eq:s1s2}, and then proceed to compute $R_1$ and $R_2$ according to~\eqref{eq:r1r2}, where $\sigma$ is replaced by $\mathcal{F}(\sigma)$, $\sigma^{-1}(j)$, $j\in[n]$ is replaced by $\sigma^{-1}(v_j)$, and $v_j$ is the $j$th smallest value in $V$. Represent $R_1$ and $R_2$ using a sequence of $t'$ symbols $r_1,\ldots,r_{t'}$ over an alphabet size $n$. Let $r_{t'+1}=\oplus_{i\in[t']}r_i$, where $\oplus$ is the sum modulo $n$.
    \item[\textbf{(3)}] Insert $n-t(t'+1)+(t+1)i$, $i\in[t'+1]$ after the $r_i$th (or $n-r_i$ if $r_i=0$) symbol in $\mathcal{F}(\sigma)$. If $r_i$ and $r_j$, $i<j,$ have the same value, insert $n+t'+1-(t+1)(t'+1)+(t+1)j$ after $n+t'+1-(t+1)(t'+1)+(t+1)i$, where $n+t'+1-(t+1)(t'+1)+(t+1)i$ is inserted after the $r_i$th symbol in $\mathcal{F}(\sigma)$. 
\end{itemize}
Let the output of the encoding procedure be $\mathcal{E}_b(\sigma)$. The decoding procedure is the reverse of the encoding procedure, explained in what follows.

\textbf{Decoding:}
\begin{itemize}
    \item[\textbf{(1)}] Given an erroneous permutation $\mathcal{E}^{e}_b(\sigma)$ of $\mathcal{E}_b(\sigma)$, if none of the redundant symbols with values $n-t(t'+1)+(t+1)i$, $i\in[t'+1]$ are missing or repeated, 
    let $r_i$, $i\in[t'+1]$ be the number of symbols with values among $V$ and placed at positions ahead of the symbol with value $n-t(t'+1)+(t+1)i$, i.e.,
    \begin{align}\label{eq:ri}
r_i=|\{j:j<a,\mathcal{E}^{e}_b(\sigma)(a)=(n-t(t'+1)+(t+1)i)\mathcal{E}^{e}_b(\sigma)(j)\in V\}|        
    \end{align}
    is the number of symbols in $\mathcal{E}^{e}_b(\sigma)$ that precede $n-t(t'+1)+(t+1)i$. Otherwise, let $n-t(t'+1)+(t+1)i$ be the missing or repeated symbol value for some $i\in[t'+1]$ and let $j_1,j_2,\ldots,j_{t+1}$ be the positions of the repeated symbols in $\mathcal{E}^{e}_b(\sigma)$. Find the unique position $j_{s}$ among $s\in[t+1]$, such that %$\mathcal{E}'_b(\sigma)(j_s-1)\notin \cup_{j>i}\{n-k(k'+1)+(k+1)j\}$, $\mathcal{E}'_b(\sigma)(j_s+1)\notin \cup_{j<i}\{n-k(k'+1)+(k+1)j\}$, and 
    if $\mathcal{E}^{e}_b(\sigma)(j_s)=n-t(t'+1)+(t+1)i$, then the sum of values of $r_i$ modulo $n$, where $r_i$ is given by~\eqref{eq:ri}, $i\in[t+1]$, equals $0$. Then, let $r_i$ be the corresponding number given by~\eqref{eq:ri}.
    \item[\textbf{(2)}] Let $\hat{\mathcal{F}}^{e}_b(\sigma)$ be the subsequence of $\mathcal{E}^{e}_b(\sigma)$ obtained by removing symbols with values $n-t(t'+1)+(t+1)i$, $i\in[t+1]$, where the symbol $\mathcal{E}^{e}_b(\sigma)(j_s)=n-t(t'+1)+(t+1)i$ obtained from Step $(1)$ is removed as well. Declare erasures of $\sigma^{i,1}$ and $\sigma^{i,2}$ in $S_1$ and $S_2$, where $\sigma^{i,1},\sigma^{i,2}, S_1$, and $S_2$ are defined in~\eqref{eq:s1s2} and~\eqref{eq:r1r2}, if at least one value among the $2(i-1)t+1$th,$\ldots,2it$th smallest or the $t+2(i-1)t+1$th,$\ldots,t+2it$th smallest entries in $V$ is missing in $\mathcal{E}^{e}_b(\sigma)$, respectively. Note that to compute $S_1$ and $S_2$ in~\eqref{eq:s1s2}, we replace $\sigma^{-1}(j)$, $j\in[n]$, by $\sigma^{-1}(v_j)$, where $v_j$ is the $j$th smallest number in $V$.
    \item[\textbf{(3)}] Find at least one of $S_1$ and $S_2$ that has a single erasure of $\sigma^{i,1}$ or $\sigma^{i,1}$, respectively. Suppose $S_1$ has a single erasure $\sigma^{i,1}$; then, it can be corrected with the help of $R_1$ defined in~\eqref{eq:r1r2}, which is part of $(r_1,\ldots,r_{t'})$ retrieved from Step (1). 
    Once $\sigma^{i,1}$ is recovered, we correct the burst stuck-at error as follows. Let $i_1<\ldots<i_{2t}$ be the positions of symbols that are in $\sigma^{i,1}$, which can be determined since the positions of other $\sigma^{j,1}$, $j\in[n]\backslash\{i\}$ can be determined as well. Then, let $\hat{\mathcal{F}}'_b(\sigma)(i_\ell)=v_{2(i-1)(t+1)+\sigma^{i,1}(\ell)}$ for $\ell\in[2t]$.
    \item[\textbf{(4)}] Recover $\sigma$ from $\hat{\mathcal{F}}^{e}_b(\sigma)$ by letting $\sigma(j)=i$ if $\hat{\mathcal{F}}^{e}_b(\sigma)(j)=v_i$. 
\end{itemize}
We now prove the correctness of the encoding/decoding procedures. We first show that $(r_1,\ldots,r_{t'})=(R_1,R_2)$ via the following lemma.
\begin{lemma}\label{lem:uniqueri}
There is a unique position $j_s$ for some $s\in [t+1]$ in Step (1) in the decoding procedure such that %$\mathcal{E}'_b(\sigma)(j_s-1)\notin \cup_{j>i}\{n-k(k'+1)+(k+1)j\}$, $\mathcal{E}'_b(\sigma)(j_s+1)\notin \cup_{j<i}\{n-k(k'+1)+(k+1)j\}$, and 
by letting $\mathcal{E}^{e}_b(\sigma)(j_s)=n-t(t'+1)+(t+1)i$ and letting $r_i$ be given by~\eqref{eq:ri}, $i\in[t+1]$, the sum of the $r_i$ values modulo $n$ equals $0$.      
\end{lemma}
\begin{proof}
Note that the burst stuck-at error affects at most one redundant symbol among $n-t(t'+1)+(t+1)i$, $i\in[t'+1]$. 
By Step (2) and Step (3) of the encoding procedure, the position of the symbol $n-t(t'+1)+(t+1)i$ in the encoding satisfies $\sum^{t'+1}_{i=1}r_i\equiv 0\bmod n$. We now show that different choices of $s\in[t+1]$ result in different modulo sum values $\sum^{t'+1}_{i=1}r_i \bmod n$. Let $a_s=\sum^{t'+1}_{i=1}r_i\equiv 0\bmod n$, $s\in[t+1],$ when $j_s$ is selected. Note that for $j_{s_1}>j_{s_2}$, we have
\begin{align*}
a_{s_1}-a_{s_2}\equiv &|\{j:j<j_{s_1},j>_{s_2},\mathcal{E}^{e}_b(\sigma)(j)\in V\}|+1\\
&+ |\{j:j<j_{s_1},j>_{s_2},\mathcal{E}^{e}_b(\sigma)(j)\in ([n+t'+1]\backslash V)\}|\\
\equiv &j_{s_1}-j_{s_2}\bmod n. 
\end{align*}
Hence, $a_s$ are different for different choices of $s\in [t+1]$.
\end{proof}
From Lemma~\ref{lem:uniqueri}, we know that $(r_1,\ldots,r_{t})$ can be correctly recovered from $\mathcal{E}^{e}_b$ during Step (1) of decoding. From Lemma~\ref{lem:consecutivestuck}, an erasure of either $\sigma^{i_1,1}$ for some $i_1\in [\lceil\frac{n}{2(t+1)}\rceil]$ or $\sigma^{i_2,2}$ for some $i_2\in [\lceil\frac{n}{2(t+1)}\rceil]$ in $S_1$ or $S_2$, respectively, can be identified such that $\sigma^{i_1,1}$ or $\sigma^{i_2,2}$ is the unique erasure in $S_1$ or $S_2$, respectively. In addition, the location of the symbols onto which $\sigma^{i_1,1}$ or $\sigma^{i_2,2}$ is projected can be deduced. Then, from the redundancy $(r_1,\ldots,r_{t})$ recovered in Step (1), $\sigma^{i_1,1}$ or $\sigma^{i_2,2}$ can be reconstructed, and in turn, from them one can infer the values of the repeated symbols in $\hat{\mathcal{F}}^{e}_b(\sigma)$ of Step (3) of decoding. Thus, one can recover $\mathcal{F}_b(\sigma)$. Finally, $\sigma$ can be recovered from the correctly decoded $\mathcal{F}_b(\sigma)$ in Step (1) of encoding. 

\subsection{The stuck-at errors model under rank modulation}
We now consider stuck-at errors for cases where the symbol values in the erroneous permutation only depend on the rankings of the average tail lengths (no quantization). Consider Example~\ref{ex:3} where the information is encoded by the permutation $\sigma = (9,1,4,2,5,8,3,6,7)$. We consider the inverse $\sigma^{-1}=(\sigma^{-1}(1),\ldots,\sigma^{-1}(9))=(2,4,7,3,5,8,9,6,1)$. It can be shown that $\sigma_e^{-1}$ can be obtained from $\sigma^{-1}$ by a symbol deletion and a symbol erasure where the set of values of the erased symbol and the deleted symbol are known (but which value corresponds to an erasure or deletion is ambiguous). Moreover, the positions of the erasure and the deletion have a difference at most $t=3$. In the example, $\sigma_e^{-1}=(2,?,6,3,8,9,6,1)$, where the question mark in $\sigma_e^{-1}(2)$ can be either $4$ or $5$. It can be seen that $\sigma_e^{-1}$ can be obtained from $\sigma^{-1}$ by deleting the symbol $5$ and erasing the symbol $4$.  
To correct an erasure in $\sigma^{-1}$ the value of which has two possibilities and an additional deletion, we use a set of parity checks that will be able to: (1) Find the correct value of the erased symbol; (2) Correct the deletion when the value of the erased symbol is fixed. For the first setting, we consider parity-checks based on a binary vector indicating the ascending or descending order of symbols, given by $(1,1,1,0,1,1,1,0,0)$ for $\sigma$, as well as the Lehmer encoding (defined in Section \ref{sec:kstuckerrors}) $\mathcal{L}(\sigma)=(0,1,1,2,1,1,4,2,2)$ of $\sigma$. Details will be provided later.

To encode parity checks into symbols of a permutation, we follow a similar approach to the one described in Section~\ref{sec:kstuckerrors} and Section~\ref{sec:burststuckerror} and use the positions of redundant symbols to encode the parity-checks. However, the ideas behind how parity checks are encoded into positions of redundant symbols and how they are decoded are more more involved.   
%The first parity check is the sum of all entries of Lehmer encoding modulo $k$. The second is the position-weighted sum of all entries of the Lehmer encoding modulo $k^2$. Hence, the two parity checks for $\sigma$ are given by $(2,5)$. It can be verified that when the question mark in $\sigma'^{-1}(2)$ equals $5$, 
We now provide a detailed description of the encoding and decoding process.
\begin{theorem}
For any message given in the form of a permutation $\sigma$ of length $n\ge t+12$, there is an encoding $\mathcal{E}_b:\mathcal{S}_n\rightarrow\mathcal{S}_{n+t'+1}$ that maps $\sigma$ to a permutation $\mathcal{E}_r(\sigma)$ of length $n+t'+1$ such that $\prod^n_{j=n-t'+1}j \ge 2(t+2)(2t+1)t^2$. Moreover, $\mathcal{E}_r(\sigma)$ can be corrected from a stuck-at symbol error described in~\eqref{eq:rankmod}.    
\end{theorem}
\begin{remark}
Note that for each erroneous permutation, there are at least $t$ choices for the original, uncorrupted permutation. Hence, the redundancy of the code is at least $\log t$. 
\end{remark}
For a permutation or a vector $\sigma\in [n]^n$, let 
\begin{align}\label{eq:sigmainverse}
\sigma^{-1}=(\sigma^{-1}(1),\ldots,\sigma^{-1}(n))    
\end{align}
be the inverse vector of $\sigma$, where $\sigma^{-1}(i)=?$ if there are repeated symbols of value $i$ in $\sigma$. Note that there is a one-to-one mapping between $\sigma$ and $\sigma^{-1}$. We consider error correction for the inverse $\sigma^{-1}$. The following lemma shows how a stuck-at symbol error affects $\sigma^{-1}$.
\begin{lemma}\label{lem:rankstuck}
Let $\sigma_e$ be the erroneous version of $\sigma$ described in~\eqref{eq:rankmod}. Let $\sigma_e(i)=a$ and $\sigma_e(i')=a$ be the repeated symbols in $\sigma_e$. Then $\sigma^{-1}\in [n]^{n-1}$ can be obtained from $\sigma_e^{-1}$ by letting $\sigma_e^{-1}(a)=i$ or $\sigma_e^{-1}(a)=i'$ and inserting a symbol of value $i'$ or $i$ after $\sigma_e^{-1}(a+t_1-1)$ or $\sigma_e^{-1}(a+t_2-1)$ for some $1\le t_1\le t$ or $1\le t_2\le t$, respectively. 
\end{lemma}
\begin{proof}
Since $\sigma_e$ have repeated symbols $\sigma_e(i)=\sigma_e(i')=a$, the stuck-at error occurs at $\sigma(i)$ or $\sigma(i')$. If the stuck-at error occurs at $\sigma(i)$, we have
\begin{align}\label{eq:rankmodinverse}
    \sigma_e^{-1}(j)=\begin{cases}
        \sigma^{-1}(j+1), &\mbox{for $j\ge \sigma(i)$},\\
        \mbox{?,}&\mbox{if $j=a$},\\
        \sigma^{-1}(j),&\mbox{else},
    \end{cases},
\end{align}
which becomes $\sigma^{-1}$ by letting $\sigma_e^{-1}(a)=i'$ and inserting a symbol with value $\sigma^{-1}(\sigma(i))=i$ after the $(\sigma(i)-1)$th symbol in $\sigma_e^{-1}$. In addition, we have  $1\le\sigma(i)-a\le t$.
Similarly, if the stuck-at error occurs at $\sigma(i')$ then $\sigma_e^{-1}$ becomes $\sigma^{-1}$ by letting $\sigma_e^{-1}(a)=i$ and inserting a symbol with value $\sigma^{-1}(\sigma(i'))=i'$ after the $(\sigma(i')-1)$th symbol in $\sigma_e^{-1}$, where $1\le \sigma(i')-a\le t$. This proves the claim.
\end{proof}
From Lemma \ref{lem:rankstuck}, it suffices to determine which of the two values between $i$ or $i'$ is the value of the erased symbol and correct the deletion of the symbol of the other value $i'$ or $i$, respectively. To this end, we consider the following binary vector $\boldsymbol{b}(\sigma^{-1})$ that indicates the ascending/descending order of symbols in $\sigma^{-1}$:
\begin{align*}
    \boldsymbol{b}(\sigma^{-1})(i)=\begin{cases}
        1,&\mbox{if $\sigma^{-1}(i)>\sigma^{-1}(i-1)$}\\
        0,&\mbox{else}
    \end{cases}.
\end{align*}
In addition, it is assumed that $\boldsymbol{b}(\sigma^{-1})(1)=1$. The following observation can be verified.
\begin{proposition}\label{obs:1}
A symbol deletion in $\sigma^{-1}(i)$ results in a bit deletion in $\boldsymbol{b}(\sigma^{-1})(i)$ or $\boldsymbol{b}(\sigma^{-1})(i+1)$. Moreover, a symbol substitution in $\sigma^{-1}(i)$ results in one of the following: (1) $(\boldsymbol{b}(\sigma^{-1})(i),$ $\boldsymbol{b}(\sigma^{-1})(i+1))$ changed from $(1,0)$ to $(0,1)$ or vice versa. (2) One of $\boldsymbol{b}(\sigma^{-1})(i)$ and $\boldsymbol{b}(\sigma^{-1})(i+1)$ flipped. (3) No changes in $\boldsymbol{b}(\sigma^{-1})$.
\end{proposition}
Based on Proposition~\ref{obs:1} and Lemma~\ref{lem:rankstuck}, we define the following parity-checks for $\sigma^{-1}$:
\begin{align}\label{eq:paritybylehmer}
p_1=&\sum^{n}_{j=1}\boldsymbol{b}(\sigma^{-1})(j)\bmod 2,\;\;\;
p_2=\sum^{n}_{j=1}j\boldsymbol{b}(\sigma^{-1})(j)\bmod (t+2)\nonumber\\
p_3=&\sum^{n}_{j=1}(\sum^j_{\ell=1}\ell)\boldsymbol{b}(\sigma^{-1})(j)\bmod t^2,\;\;\;
p_4= \sum^{n}_{j=1}\mathcal{L}(\sigma^{-1})(j)\bmod (2t+1),
\end{align}
where $\mathcal{L}(\sigma^{-1})$ is 
the Lehmer encoding of $\sigma^{-1}$ defined in~\eqref{eq:lehmer}. 
The following lemma shows that $(p_1,p_2,p_3,p_4)$ can be used to correct a stuck-at symbol error in $\sigma^{-1}$.
\begin{lemma}\label{lem:correctnessofparity}
Let $\sigma_e$ be the erroneous vector described by~\eqref{eq:rankmod} and let $\sigma_e(i)=\sigma_e(i')=a$ be the repeated symbols in $\sigma_e$. Then, any 
two different permutations $\sigma^{-1}_1$ and $\sigma^{-1}_2$ obtained from $\sigma_e^{-1}$ by letting $\sigma_e^{-1}(a)=j_1$ and $\sigma_e^{-1}(a)=j_2$, respectively, for some $j_1,j_2\in\{i,i'\}$, and inserting a symbol with value $\{i,i'\}\backslash\{j_1\}$ and $\{i,i'\}\backslash\{j_2\}$ after the $(a+t_1-1)$th and $(a+t_2-1)$th symbol of $\sigma_e^{-1}$, respectively, where $1\le t_1,t_2\le t$, have different parity-checks $(p_1,p_2,p_3,p_4)$.% for $\sigma^{-1}_1$ are different from those for $\sigma^{-1}_2$.
\end{lemma}
\begin{proof}
Let $\sigma^{-1}_{e1}$ and $\sigma^{-1}_{e2}$ be the vectors obtained from $\sigma_e^{-1}$ by letting $\sigma_e^{-1}(a)=j_1$ and $\sigma_e^{-1}(a)=j_2$, respectively, for some $j_1,j_2\in\{i,i'\}$. Then from Proposition \ref{obs:1}, $\boldsymbol{b}(\sigma^{-1}_{e1})$ and $\boldsymbol{b}(\sigma^{-1}_{e2})$ can be obtained by deleting  $\boldsymbol{b}(\sigma^{-1}_1)(a+t_1)$ or $\boldsymbol{b}(\sigma^{-1}_1)(a+t_1+1)$ from 
$\boldsymbol{b}(\sigma^{-1}_1)$ 
and $\boldsymbol{b}(\sigma^{-1})(a+t_2)$ or $\boldsymbol{b}(\sigma^{-1})(a+t_2+1)$ from  $\boldsymbol{b}(\sigma^{-1}_2)$, respectively, where $1\le t_1,t_2\le t$. Moreover, we have one of the following: (1) $\boldsymbol{b}(\sigma^{-1}_{e1})$ and $\boldsymbol{b}(\sigma^{-1}_{e2})$ differ only in the positions $a$ and $a+1$ such that either $(\boldsymbol{b}(\sigma^{-1}_{e1})(a),\boldsymbol{b}(\sigma^{-1}_{e1})(a+1))=(0,1)$ or $(\boldsymbol{b}(\sigma^{-1}_{e1})(a),\boldsymbol{b}(\sigma^{-1}_{e1})(a+1))=(1,0)$; (2) $\boldsymbol{b}(\sigma^{-1}_{e1})$ and $\boldsymbol{b}(\sigma^{-1}_{e2})$ differ only in position $a$ or $a+1$; (3) $\boldsymbol{b}(\sigma^{-1}_{e1})$ and $\boldsymbol{b}(\sigma^{-1}_{e2})$ are equal. In what follows, we show that if the parity checks $(p_1,p_2,p_3)$ for $\sigma^{-1}_1$ and $\sigma^{-1}_2$ are equal, then $\boldsymbol{b}(\sigma^{-1}_1)=\boldsymbol{b}(\sigma^{-1}_2)$, for all three cases. 

We start with case (3). As mentioned above, $\boldsymbol{b}(\sigma^{-1}_{e1})$ and $\boldsymbol{b}(\sigma^{-1}_{e2})$ are obtained from $\boldsymbol{b}(\sigma^{-1}_{1})$ and $\boldsymbol{b}(\sigma^{-1}_{2})$, respectively, after a single deletion. If $\boldsymbol{b}(\sigma^{-1}_{e1})=\boldsymbol{b}(\sigma^{-1}_{e2})$, $\boldsymbol{b}(\sigma^{-1}_{1})$ and $\boldsymbol{b}(\sigma^{-1}_{2})$ share a common subsequence of length $n-1$. It was shown 
 in~\cite{levenshtein1966binary} that if $\boldsymbol{b}(\sigma^{-1}_{1})$ and $\boldsymbol{b}(\sigma^{-1}_{2})$ share a common subsequence of length $n-1$, the Varshamov-Tenengolt parity check, described by $p_2$ in \eqref{eq:paritybylehmer}, of $\boldsymbol{b}(\sigma^{-1}_1)$ is different from that of $\boldsymbol{b}(\sigma^{-1}_{2})$. Here we briefly illustrate the proof. Note that when the parity-checks $p_1,p_2,$ and $p_3$ of $\boldsymbol{b}(\sigma^{-1}_1)$ and $\boldsymbol{b}(\sigma^{-1}_2)$ are the same, they remain the same when $\boldsymbol{b}(\sigma^{-1}_1)$ and $\boldsymbol{b}(\sigma^{-1}_2)$ flip all their bits. Hence, without loss of generality, we can assume that $\boldsymbol{b}(\sigma^{-1}_1)$ and $\boldsymbol{b}(\sigma^{-1}_2)$ are obtained from $\boldsymbol{b}(\sigma^{-1}_{e1})$ by inserting bit $0$ at positions $a+t'_1$ and $a+t'_2$, respectively, where $1\le t'_1,t'_2\le t+1$. Then 
\begin{align}\label{eq:r2equal}
&\sum^{n}_{j=1}j\boldsymbol{b}(\sigma^{-1}_1)(j)-\sum^{n}_{j=1}j\boldsymbol{b}(\sigma^{-1}_2)(j)\nonumber\\
\equiv& 
 |\{j:j\ge a+t'_1,j\le a+t+1,\boldsymbol{b}(\sigma^{-1}_1)(j)=1\}|\nonumber\\&-|\{j:j\ge a+t'_2,j\le a+t+1,\boldsymbol{b}(\sigma^{-1}_2)(j)=1\}|
 \bmod (t+2).  
\end{align}
Since $0\le|\{j:j\ge a+t'_1,j\le a+t+1,\boldsymbol{b}(\sigma^{-1}_1)(j)=1\}|,|\{j:j\ge a+t'_2,j\le a+t+1,\boldsymbol{b}(\sigma^{-1}_2)(j)=1\}|\le t+1$, we have
$$|\{j:j\ge a+t'_1,j\le a+t+1,\boldsymbol{b}(\sigma^{-1}_1)(j)=1\}|=|\{j:j\ge a+t'_2,j\le a+t+1,\boldsymbol{b}(\sigma^{-1}_2)(j)=1\}|,$$
which implies that the $0$ bit is inserted in the same run or consecutive bits of $0$'s in $\boldsymbol{b}(\sigma^{-1}_{e1})$ to obtain $\boldsymbol{b}(\sigma^{-1}_1)$ or $\boldsymbol{b}(\sigma^{-1}_2)$, respectively, implying that $\boldsymbol{b}(\sigma^{-1}_1)=\boldsymbol{b}(\sigma^{-1}_2)$. 

We now prove that $\boldsymbol{b}(\sigma^{-1}_1)=\boldsymbol{b}(\sigma^{-1}_2)$ for case (1). 
Since the parity checks $p_1$ for $\boldsymbol{b}(\sigma^{-1}_1)$ and $\boldsymbol{b}(\sigma^{-1}_2)$ are the same, $\boldsymbol{b}(\sigma^{-1}_1)$ and $\boldsymbol{b}(\sigma^{-1}_2)$ can be obtained from $\boldsymbol{b}(\sigma^{-1}_{e1})$ and $\boldsymbol{b}(\sigma^{-1}_{e2})$ by inserting a $0$ bit or $1$ bit at positions $a+t'_1$ and $a+t'_2$, respectively, for some $1\le t'_1,t'_2\le t+1$. Again, without loss of generality, we assume that the inserted bits are $0$-bits to obtain $\boldsymbol{b}(\sigma^{-1}_1)$ and $\boldsymbol{b}(\sigma^{-1}_2)$, respectively. Moreover, we assume that $(\boldsymbol{b}(\sigma^{-1}_{e1})(a),\boldsymbol{b}(\sigma^{-1}_{e1})(a+1))=(0,1)$ and $(\boldsymbol{b}(\sigma^{-1}_{e2})(a),\boldsymbol{b}(\sigma^{-1}_{e2})(a+1))=(1,0)$. 
Then, similar to previous case, we have 
\begin{align*}
&|\{j:j\ge a+t'_1,j\le a+t+1,\boldsymbol{b}(\sigma^{-1}_1)(j)=1\}|+1\\
=&|\{j:j\ge a+t'_2,j\le a+t+1,\boldsymbol{b}(\sigma^{-1}_2)(j)=1\}|,
\end{align*}
which implies
\begin{align}\label{eq:setdiff}
&\{j:j\ge a+t'_2,j\le a+t+1,\boldsymbol{b}(\sigma^{-1}_2)(j)=1\}\nonumber\\
=&\{j:j\ge a+t'_1,j\le a+t+1,\boldsymbol{b}(\sigma^{-1}_1)(j)=1\}\cup \{j_1\},    
\end{align}
for some $j_1\in \{a+1,\ldots,a+t+1\}$.
Then, we have
\begin{align*}
 &\sum^{n}_{j=1}(\sum^j_{\ell=1}\ell)\boldsymbol{b}(\sigma^{-1}_1)(j)-\sum^{n}_{j=1}(\sum^j_{\ell=1}\ell)\boldsymbol{b}(\sigma^{-1}_2)(j) \\
 =&a+1+\sum_{j:j\ge a+t'_1,j\le a+t+1,\boldsymbol{b}(\sigma^{-1}_1)(j)=1}(j+1)-\sum_{j:j\ge a+t'_2,j\le a+t+1,\boldsymbol{b}(\sigma^{-1}_1)(j)=1}(j+1)\\
 =&a+1-j_1-1.
\end{align*}
Recall that $1\le j_1\le t+1$. Hence,
\begin{align}\label{eq:r3necase1}
 \sum^{n}_{j=1}(\sum^j_{\ell=1}\ell)\boldsymbol{b}(\sigma^{-1}_1)(j)\not\equiv\sum^{n}_{j=1}(\sum^j_{\ell=1}\ell)\boldsymbol{b}(\sigma^{-1}_2)(j)\bmod t^2,  
\end{align}
if $\boldsymbol{b}(\sigma^{-1}_1)\ne\boldsymbol{b}(\sigma^{-1}_2)$, contradicting the assumption that $p_3$ is equal for $\boldsymbol{b}(\sigma^{-1}_1)$ and $\boldsymbol{b}(\sigma^{-1}_2)$. 
%Let $A=\{\sigma'^{-1}(a+1),\ldots,\sigma'^{-1}(a+t_1-1)\}$ and $B=\{\sigma'^{-1}(a+t_1),\ldots,\sigma'^{-1}(a+t_2-1)\}$.

We now show that $\boldsymbol{b}(\sigma^{-1}_1)=\boldsymbol{b}(\sigma^{-1}_2)$ for case (2). Without loss of generality, assume that $\boldsymbol{b}(\sigma^{-1}_{e1})$ and $\boldsymbol{b}(\sigma^{-1}_{e2})$ differ in $a'\in\{a,a+1\}$ such that $\boldsymbol{b}(\sigma^{-1}_{e1})(a')=1$ and $\boldsymbol{b}(\sigma^{-1}_{e2})(a')=0$. Then, since the parity checks $p_1$ for $\boldsymbol{b}(\sigma^{-1}_1)$ and $\boldsymbol{b}(\sigma^{-1}_2)$ are equal, we have that $\boldsymbol{b}(\sigma^{-1}_1)$ and $\boldsymbol{b}(\sigma^{-1}_2)$ can be obtained from $\boldsymbol{b}(\sigma^{-1}_{e1})$ and $\boldsymbol{b}(\sigma^{-1}_{e2})$ by inserting a $0$ bit and $1$ bit at positions $a+t'_1$ and $a+t'_2$, respectively, for some $1\le t'_1,t'_2\le t+1$. We consequently have 
\begin{align*}   &\sum^{n}_{j=1}j\boldsymbol{b}(\sigma^{-1}_1)(j)-\sum^{n}_{j=1}j\boldsymbol{b}(\sigma^{-1}_2)(j)\\
=&|\{j:j\ge a+t'_1,j\le a+t+1,\boldsymbol{b}(\sigma^{-1}_1)(j)=1\}|\\
&-|\{j:j\ge a+t'_2,j\le a+t+1,\boldsymbol{b}(\sigma^{-1}_2)(j)=1\}|-(a+t'_2-a').
\end{align*}
When the parity checks $p_2$ for $\boldsymbol{b}(\sigma^{-1}_1)$ and $\boldsymbol{b}(\sigma^{-1}_2)$ are equal, we have
\begin{align*}
&\{j:j\ge a+t'_1,j\le a+t+1,\boldsymbol{b}(\sigma^{-1}_1)(j)=1\}\\
=&\{j:j\ge a+t'_2,j\le a+t+1,\boldsymbol{b}(\sigma^{-1}_2)(j)=1\}\cup \{j_1,\ldots,j_{a+t'_2-a'}\}    
\end{align*}
for some $j_1,\ldots,j_{a+t'_2-a'}\in \{a+1,\ldots, a+t+1\}$ that are different. Then,
\begin{align*}
 &\sum^{n}_{j=1}(\sum^j_{\ell=1}\ell)\boldsymbol{b}(\sigma^{-1}_1)(j)-\sum^{n}_{j=1}(\sum^j_{\ell=1}\ell)\boldsymbol{b}(\sigma^{-1}_2)(j) \\
 =&\sum_{j:j\ge a+t'_1,j\le a+t+1,\boldsymbol{b}(\sigma^{-1}_1)(j)=1}(j+1)-\sum_{j:j\ge a+t'_2,j\le a+t+1,\boldsymbol{b}(\sigma^{-1}_1)(j)=1}(j+1)-\sum^{a+t'_2-a'}_{\ell=a'+1}\ell\\
 =&\sum^{a+t'_2-a'}_{\ell=1}(j_\ell+1)-\sum^{a+t'_2-a'}_{\ell=a'+1}\ell,
\end{align*}
which is greater than $0$ and smaller than $(\frac{t+1}{2})^2\le t^2$. Hence, we have \eqref{eq:r3necase1}, which contradicts the assumption that the parity-checks $p_3$ for $\boldsymbol{b}(\sigma^{-1}_1)$ and $\boldsymbol{b}(\sigma^{-1}_2)$ are equal.

Next, we show that if $\boldsymbol{b}(\sigma^{-1}_1)=\boldsymbol{b}(\sigma^{-1}_2)$ and the parity check $p_4$ for $\boldsymbol{b}(\sigma^{-1}_1)$ and $\boldsymbol{b}(\sigma^{-1}_2)$ are equal, then we have $\sigma^{-1}_1=\sigma^{-1}_2$. If $\sigma^{-1}_1(a)=\sigma^{-1}_2(a)$, we have that $\sigma^{-1}_1$ and $\sigma^{-1}_2$ are obtained from $\sigma^{-1}_{e1}$ by inserting a symbol with the same value at positions $a+t_1$ and $a+t_2$, respectively, such that $\boldsymbol{b}(\sigma^{-1}_1)=\boldsymbol{b}(\sigma^{-1}_2)$. This implies that the symbol is inserted in the same increasing run or decreasing run in $\sigma^{-1}_{e1}$ to obtain $\sigma^{-1}_1$ and $\sigma^{-1}_2$, respectively, where an increasing or decreasing run in a vector $\boldsymbol{c}=(c(1),\ldots,c(n))$ is a subsequence of consecutive symbols $(c(i+1),\ldots,c(i+j))$ such that $c(i+1)<\ldots<c(i+j)$ or $c(i+1)>\ldots>c(i+j)$, respectively. Hence, $\sigma^{-1}_1$ and $\sigma^{-1}_2$ are equal. On the other hand, if $\sigma^{-1}_1(a)=j_1$ and $\sigma^{-1}_2(a)=j_2$ are different, then $\sigma^{-1}_1$ and $\sigma^{-1}_2$ are obtained from $\sigma^{-1}_{e1}$ and $\sigma^{-1}_{e2}$ by inserting a symbol with values $j_2$ and $j_1$ at positions $a+t_1$ and $a+t_2$, respectively. Moreover, similarly as above, from $\boldsymbol{b}(\sigma^{-1}_1)=\boldsymbol{b}(\sigma^{-1}_2)$ we have that the symbols $j_2$ and $j_1$ are inserted in the same increasing run or decreasing run in $\sigma^{-1}_{e1}$ and $\sigma^{-1}_{e2}$ to obtain $\sigma^{-1}_1$ and $\sigma^{-1}_2$, respectively. Without loss of generality, let $j_2\ge j_1$, then,
\begin{align*}
 &\sum^{n}_{j=1}\mathcal{L}(\sigma^{-1}_2)(j)-\sum^{n}_{j=1}\mathcal{L}(\sigma^{-1}_1)(j)\\
 =&|\{j:j\ge a+1,j\le a+t_1-1,\sigma^{-1}_{e1}(j)<j_1\}|+|\{j:j\ge a+1,j\le a+t_1-1,\sigma^{-1}_{e1}(j)>j_2\}|\\
 &+1-|\{j:j\ge a+1,j\le a+t_2-1,\sigma^{-1}_{e2}(j)<j_2\}|\\
 &-|\{j:j\ge a+1,j\le a+t_2-1,\sigma^{-1}_{e2}(j)>j_1\}|.
\end{align*}
If $j_2$ and $j_1$ are inserted in an increasing run in $\sigma^{-1}_{e1}$ and $\sigma^{-1}_{e2}$, respectively, to obtain $\sigma^{-1}_1$ and $\sigma^{-1}_2$, then we have that $t_1< t_2$. Since $\sigma^{-1}_{e1}(j)=\sigma^{-1}_{e2}(j)$ for $a+1\le j\le a+t_1-1$, then,  
\begin{align*}
    &|\{j:j\ge a+1,j\le a+t_1-1,\sigma^{-1}_{e1}(j)<j_1\}|+|\{j:j\ge a+1,j\le a+t_1-1,\sigma^{-1}_{e1}(j)>j_2\}|\\
    &+1-|\{j:j\ge a+1,j\le a+t_2-1,\sigma^{-1}_{e2}(j)<j_2\}|\\
    &-|\{j:j\ge a+1,j\le a+t_2-1,\sigma^{-1}_{e2}(j)>j_1\}|\\
    =&2|\{j:j\ge a+1,j\le a+t_1-1,\sigma^{-1}_{e1}(j)<j_1,\sigma^{-1}_{e1}(j)>j_2\}|+1,
\end{align*}
which is a value between $1$ and $2t+1$. Hence, 
\begin{align}\label{eq:r4ne}
 &\sum^{n}_{j=1}\mathcal{L}(\sigma^{-1}_2)(j)\not\equiv \sum^{n}_{j=1}\mathcal{L}(\sigma^{-1}_1)(j)\bmod (2t+1).
\end{align}
Similarly, \eqref{eq:r4ne} holds when  $j_2$ and $j_1$ are inserted in an increasing run in $\sigma^{-1}_{e1}$ and $\sigma^{-1}_{e2}$, respectively. Hence, we have that $\sigma^{-1}_1=\sigma^{-1}_2$ whenever the two inverse permutations have the same parity-checks $(p_1,p_2,p_3,p_4)$.
\end{proof}
Lemma \ref{lem:rankstuck} shows that given $\sigma_e$ described by \eqref{eq:rankmod}, $\sigma^{-1}$ and thus $\sigma$ can be recovered with the help of parity checks $(p_1,p_2,p_3,p_4)$ of $\sigma$. In the following, we show how to use redundant symbols to encode $(p_1,p_2,p_3,p_4)$. 
Same as in Section \ref{sec:burststuckerror}, we do not make any assumption on $m$. 
We follow a similar manner to the one in Section \ref{sec:kstuckerrors} and Section \ref{sec:burststuckerror}, where the positions of redundant symbols are used to encode $(p_1,\ldots,p_4)$. However, the encoding from $(p_1,\ldots,p_4)$ to positions of redundant symbols is different from that in Section \ref{sec:burststuckerror}. 

Before presenting the encoding and decoding procedures, we define a useful mapping.
\begin{proposition}\label{prop:mappingp}
There exists a one-to-one mapping $\mathcal{P}$ that maps an integer $\ell \in [\prod^s_{j=s-t+1}j]$ to $t$ different symbols from an alphabet of size $s$.   
\end{proposition}
\begin{proof}
Let $\ell =\sum^{t-1}_{i=0}= a_{i+1}\cdot \prod^{s-i}_{j=s-t+1}j$. Then, we have that $a_{i}\in\{0,\ldots,s-i\}$ for $i\in[t]$. We then map $a_{1},\ldots,a_{t}$ into $t$ different integers $j_1,\ldots,j_t$ as follows. Let $j_i$ be the $(a_i+1)$th smallest integer in $[s]\backslash\{j_1,\ldots,j_{i-1}\}$. It is clear that such a mapping is invertible.    
\end{proof}
Let $(p_1,p_2,p_3,p_4)$ be represented by $t'\le\lceil\frac{\log \big(2(t+2)(2t+1)t^2\big)}{\log (n-9)}\rceil\le 5$ different symbols $(r_1,\ldots,r_{t'})$ from an alphabet of size $n-5$, which can be done using the mapping $\mathcal{P}$ in Proposition~\ref{prop:mappingp}. Note that $t'\le 5$ because $2(t+2)(2t+1)t^2\le (n-9)^5$ when $n\ge t+12$.  
Let $r'_i=r_i+5$ for $i\in[t']$. Then $6\le r'_i\le n$. We then insert $n+i$ into $\sigma$ as the $r'_i$th symbol, $i\in [t']$. Finally, we insert the symbol $n+t'+1$ into the $\sigma$ vector (the location of the insertion is described by the following lemma) and obtain a permutation $\mathcal{E}_r(\sigma)$ of length $n+t'+1$ such that $\sum^{t'+1}_{i=1}\mathcal{E}_r(\sigma)^{-1}(n+i)\equiv 0 \bmod (n+1)$. The following lemma shows that such an insertion of $n+t'+1$ is always possible.
\begin{lemma}\label{lem:insertlastsymbol}
For any permutation $\sigma\in\mathcal{S}_{n+t'}$, it is possible to insert a symbol $n+t'+1$ into $\sigma$ to obtain a new permutation $\sigma'$ such that $\sum^{t'+1}_{i=1}\sigma'^{-1}(n+i)\equiv 0 \bmod (n+1)$.    
\end{lemma}
\begin{proof}
Note that
\begin{align*}
&\sum^{t'+1}_{i=1}\sigma'^{-1}(n+i)-\sum^{t'}_{i=1}\sigma^{-1}(n+i)  \\
=&\sigma'^{-1}(n+t'+1)+|\{j:j\ge \sigma'^{-1}(n+t'+1),\sigma(j)\in \{n+1,\ldots,n+t'\}\}|,
\end{align*}
which increases by at least $0$ and at most $1$ as $\sigma'^{-1}(n+t'+1)$ increases by $1$. Note that when $\sigma'^{-1}(n+t'+1)=1$, we have $\sum^{t'+1}_{i=1}\sigma'^{-1}(n+i)-\sum^{t'}_{i=1}\sigma^{-1}(n+i)=t'+1$ and when $\sigma'^{-1}(n+t'+1)=n+t'+1$, we have $\sum^{t'+1}_{i=1}\sigma'^{-1}(n+i)-\sum^{t'}_{i=1}\sigma^{-1}(n+i)=n+t'+1$. Hence, there always exists a choice of $\sigma'^{-1}(n+t'+1)$ in $[n+t'+1]$ such that $\sum^{t'+1}_{i=1}\sigma'^{-1}(n+i)-\sum^{t'}_{i=1}\sigma^{-1}(n+i)$ is in  $[n+t'+1]\backslash[t']$, which maps bijectively to $\mathbbm{Z}_{n+1}=\{0,\ldots,n\}$ under modulo $(n+1)$ reduction.
\end{proof}
We are now ready to present the encoding procedure.

\textbf{Encoding:}
\begin{itemize}
    \item[\textbf{(1)}] Given a permutation $\sigma\in\mathcal{S}_n$, compute the parity checks $(p_1,p_2,p_3,p_4)$ based on \eqref{eq:paritybylehmer}. Let $(p_1,p_2,p_3,p_4)$ be represented by $t'\le\lceil\frac{\log \big(2(t+2)(2t+1)t^2\big)}{\log (n-10)}\rceil\le 5$ different symbols $(r_1,r_2,\ldots,r_{t'})$ from an alphabet of size $n-5$, using the mapping $\mathcal{P}$ in Proposition~\ref{prop:mappingp}. Let $r'_i=r_i+5$ for $i\in[t']$.
    \item[\textbf{(2)}] Insert $n+i$, $i\in[t']$ into $\sigma$ such that $n+i$ is the $r'_i$th symbol in the new permutation. %If $r'_i$ and $r'_j$, $i<j$ have the same value, insert $n+j$ after $n+i$, where $n+i$ is after the $r'_i$th symbol in $\sigma$. 
    Denote the resulting permutation by $R(\sigma)$.
    \item[\textbf{(3)}] According to Lemma \ref{lem:insertlastsymbol}, insert $n+t'+1$ into $R(\sigma)$ to obtain $\mathcal{E}_r(\sigma)$ such that $\sum^{t'+1}_{i=1}\mathcal{E}_r(\sigma)^{-1}(n+i)\equiv 0 \bmod (n+1)$. 
\end{itemize}
Upon receiving an erroneous version $\mathcal{E}^{e}_r(\sigma)$ of $\mathcal{E}_r(\sigma)$, % described by \eqref{eq:rankmod}, 
we apply the following procedure. 

\textbf{Decoding:}
\begin{itemize}
    \item[\textbf{(1)}] Given an erroneous permutation $\mathcal{E}^{e}_r(\sigma)$ of $\mathcal{E}_r(\sigma)$,  %let $\mathcal{E}'_r(\sigma)(i)=\mathcal{E}'_r(\sigma)(i')=a$ be the repeated symbols in $\mathcal{E}'_r(\sigma)$.  
    compute  $\mathcal{E}'^{-1}_r(\sigma)$ based on \eqref{eq:sigmainverse}, by replacing $\sigma$ with $\mathcal{E}^{e}_r(\sigma)$.
    \item[\textbf{(2)}] Let $\mathcal{E}^{e}_r(\sigma)(i)=\mathcal{E}^{e}_r(\sigma)(i')=a$ be the repeated symbols in $\mathcal{E}^{e}_r(\sigma)$. If both $i$ and $i'$ are $>n$, remove the symbols $n+1,\ldots,n+t'$ and declare that the remaining permutation is $\sigma$. If $\min\{i,i'\}\le n$, 
    let $r=-\sum^{t'}_{j=1}\mathcal{E}'_r(\sigma)^{-1}(n+j) \bmod (n+1)$. If $i\not\equiv r\bmod (n+1) $ and $i'\not\equiv r\bmod (n+1)$, let $\mathcal{E}^{-1}_r(\sigma)(n+j)=(\mathcal{E}^{e})^{-1}_r(\sigma)(n+j-1)$ for $j\in[t'+1]$. Recover $r'_j=\mathcal{E}^{-1}_r(\sigma)(n+j)$ and $r_j=r'_j-5$ for $j\in[t']$. Let $\hat{\mathcal{E}}_r(\sigma)$ be the permutation obtained from $\mathcal{E}^{e}_r(\sigma)$ by removing symbols $n,n+1,\ldots,n+t'$. Use the redundant symbols $r_1,\ldots,r_{t'}$ to recover the parity checks $(p_1,p_2,p_3,p_4)$ of $\sigma$ and recover $\sigma^{-1}$ from $\hat{\mathcal{E}}_r(\sigma)$ and thus $\sigma$ according to Lemma~\ref{lem:correctnessofparity}.    
    If at least one of $i$ and $i'$, say $i$, satisfies $i\equiv r\bmod (n+1)$, we have either $i+n+1,i-n-1\notin [n+t'+1]$ or $i\in [t']\cup\{n+2,\ldots,n+t'+1\}$. If  $i+n+1,i-n-1\notin [n+t'+1]$, remove $\mathcal{E}^{e}_r(\sigma)(i)=a$ and the  symbols $n+1,\ldots,n+t'$ from $\mathcal{E}^{e}_r(\sigma)$ and proceed to declare the remaining permutation to be $\sigma$. %The step when $i\not\equiv r\bmod (n+1)$ and $i'\equiv r\bmod (n+1)$ follows similarly.  
    On the other hand, if $i\in [t']\cup\{n+2,\ldots,n+t'+1\}$, let $r'_j=\mathcal{E}^{-1}_r(\sigma)(n+j)$ and $r_j=r'_j-5$ for $j\in[t']$. Then recover $(p_1,p_2,p_3,p_4)$ from $r_1,\ldots,r_{t'}$. Let $\hat{\mathcal{E}}_r(\sigma)$ be the permutation obtained from $\mathcal{E}^{e}_r(\sigma)$ by removing the symbols $n,n+1,\ldots,n+t'$. Then, use $\hat{\mathcal{E}}_r(\sigma)$ and $(p_1,p_2,p_3,p_4)$ to recover $\sigma^{-1}$ and $\sigma$.  
\end{itemize}
In what follows, we prove the correctness of the decoding procedure. When $i$ and $i'$ in Step (2) of decoding are both $\geq n+1$, only redundant symbols can be erroneous. Thus removing them gives the permutation $\sigma$. In the following we focus on cases when $\min\{i,i'\}\le n$. Note that symbols 
$n+i$, $i\in[t']$, in $\mathcal{E}'_r(\sigma)$ are redundant symbols and that the sum of $n+t'+1$ redundant symbols modulo $n+1$ is $0$. Therefore, the position of the redundant symbol that is not included in the symbols $n+1,\ldots,n+t'$ in $\mathcal{E}^{e}_r(\sigma)$ is equivalent to $r$ modulo $n+1$. Hence, if the positions $i$ and $i'$ of the repeated symbols in  $\mathcal{E}^{e}_r(\sigma)$ are not equivalent to $r$ modulo $n+1$, we have that the stuck-at error does not occur among the redundant symbols. Then the symbols $n,\ldots,n+t'$ correspond to redundant symbols $n+1,\ldots,n+t'+1$ in $\mathcal{E}_r(\sigma)$ and hence can be used to recover $r'_1,\ldots,r'_{t'}$ and thus $(r_1,\ldots,r_{t'})$. Then, we can recover $p_1,\ldots,p_4$ from $(r_1,\ldots,r_{t'})$. Note that after removing the symbols $n,\ldots,n+t'$ from $\mathcal{E}^{e}_r(\sigma)$ we obtain an erroneous version $\hat{\mathcal{E}}_r(\sigma)$ of $\sigma$ described by~\eqref{eq:rankmod}. Hence, $\sigma^{-1}$ and thus $\sigma$ can be recovered from $\hat{\mathcal{E}}_r(\sigma)$ and $(p_1,p_2,p_3,p_4)$ according to Lemma~\ref{lem:correctnessofparity}. 

If one of $i$ and $i'$, say $i$, is equivalent to $r$ modulo $n+1$, then if $i+n+1,i-n-i\notin [n+t'+1]$, we have that $i$ is the position of the redundant symbol and a stuck-at error occurs at $\mathcal{E}_r(\sigma)(i)$. Thus removing $\mathcal{E}^{e}_r(\sigma)(i)=a$ and the symbols $n+1,\ldots,n+t'$ from $\mathcal{E}^{e}_r(\sigma)$ deletes the redundant symbols in $\mathcal{E}_r(\sigma)$ results in $\sigma$. On the other hand, if $i\in [t']\cup\{n+2,\ldots,n+t'+1\}$, we have that the stuck-at error occurs at symbol $n+t'+1$. Otherwise, the missing redundant symbol other than $n+1,\ldots,n+t'$ in $\mathcal{E}^{e}_r(\sigma)$ is located at a position in $[t']\cup\{n+2,\ldots,n+t'+1\}$, which contradicts the fact that the positions of redundant symbols are confined to $t'<5\le r'_j\le n$, $j\in[t']$. Therefore, the symbols $n+1,\ldots,n+t'$ in $\mathcal{E}^{e}_r(\sigma)$ correspond to symbols $n+1,\ldots,n+t'$ in $\mathcal{E}_r(\sigma)$ and thus can be used to recover $r_1,\ldots,r_{t'}$, as well as $(p_1,p_2,p_3,p_4)$. Then, removing the redundant symbols $n+1,\ldots,n+t'$ from $\mathcal{E}^{e}_r(\sigma)$ results in a erroneous version $\sigma^{e}$ of $\sigma$ that is described by~\eqref{eq:rankmod}. Hence, $\sigma^{-1}$ and $\sigma$ can be recovered from $\sigma^{e}$ and $(p_1,p_2,p_3,p_4)$.
\bibliography{scibib}
\bibliographystyle{IEEEtran}
\end{document}